\documentclass[a4paper]{article}
\pdfoutput=1

\newif\ifconf\conffalse
\newif\ifdraft\draftfalse

\usepackage[T1]{fontenc}
\usepackage[utf8]{inputenc}
\usepackage[american]{babel}

\ifconf\else
\usepackage[a4paper,centering]{geometry}
\fi
\usepackage[final]{microtype}

\usepackage{amsmath}
\usepackage{amssymb}
\usepackage{mathtools}

\ifconf
\usepackage{amsthm}
\else
\usepackage{amthm}
\fi
\usepackage{thmtools}
\usepackage{thm-restate}

\usepackage{complexity}
\newclass{\CA}{CA}
\newclass{\SCA}{SCA}
\newclass{\TISP}{TISP}
\newlang{\MAJ}{MAJ}
\newlang{\MCSP}{MCSP}
\newlang{\Merge}{Merge}
\newlang{\PAR}{PAR}
\newlang{\THR}{THR}

\let\R\relax
\usepackage{ammacros}

\DeclareMathOperator{\nndcl}{nndcl}
\DeclareMathOperator{\nndcr}{nndcr}
\DeclareMathOperator{\trta}{tt}

\newcommand{\Block}{\mathsf{Block}}
\newcommand{\B}{\mathfrak{B}}
\newcommand{\Shr}{\mathrm{S}}
\newcommand{\domega}{\omega\omega}
\newcommand{\owComm}{\mathfrak{C}_{\mathrm{ow}}}

\usepackage[ruled]{algorithm2e}
\SetCommentSty{normaltext}
\SetKw{Accept}{accept}
\SetKw{Goto}{goto}
\SetKw{Reject}{reject}

\usepackage{csquotes}
\usepackage[inline]{enumitem}
\ifconf
\newlength{\mytopsep}
\setitemize{topsep=\mytopsep}
\setenumerate{topsep=\mytopsep}
\setdescription{topsep=\mytopsep}
\fi

\ifconf\else
\usepackage[final]{hyperref}
\fi

\ifdraft
\usepackage{todonotes}
\usepackage{soul}
\colorlet{twnotecolor}{blue!9!white}

\usepackage[pagewise]{lineno}
\linenumbers
\fi

\usepackage{tikz}
\usetikzlibrary{arrows,automata,calc,positioning}

\newlength{\rowspc}
\newlength{\lblspc}

\usepackage[backend=biber,bibencoding=utf8,sortcites=true,
  maxbibnames=99]{biblatex}
\AtEveryBibitem{%
  \clearfield{doi}
  \clearfield{isbn}
  \clearfield{issn}
  \clearfield{note}
  \clearfield{url}
}

\title{Lower Bounds and Hardness Magnification for Sublinear-Time Shrinking
  Cellular Automata}
\newcommand{\myname}{Augusto Modanese}

\newcommand{\myaffil}{Karlsruhe Institute of Technology (KIT), Germany}
\newcommand{\myemail}{modanese@kit.edu}
\newcommand{\myack}{%
  I would like to thank Thomas Worsch for the helpful discussions and feedback.
}

\ifconf
\author{\myname}
\institute{\myaffil \\ \email{\myemail}}
\else
\author{\myname \\ \\ \normalsize{\myaffil} \\ \normalsize{\texttt{\myemail}}}
\date{}
\fi

\begin{document}

\maketitle

\begin{abstract}
  The minimum circuit size problem (MCSP) is a string compression problem with a
  parameter $s$ in which, given the truth table of a Boolean function over
  inputs of length $n$, one must answer whether it can be computed by a Boolean
  circuit of size at most $s(n) \ge n$.
  Recently, \citeauthor{mckay19_weak_stoc} (STOC, \citeyear{mckay19_weak_stoc})
  proved a hardness magnification result for MCSP involving (one-pass) streaming
  algorithms:
  For any reasonable $s$, if there is no $\poly(s(n))$-space streaming algorithm
  with $\poly(s(n))$ update time for $\MCSP[s]$, then $\P \neq \NP$.
  We prove an analogous result for the (provably) strictly less capable model
  of shrinking cellular automata (SCAs), which are cellular automata whose cells
  can spontaneously delete themselves.
  We show every language accepted by an SCA can also be accepted by a streaming
  algorithm of similar complexity, and we identify two different aspects in
  which SCAs are more restricted than streaming algorithms.
  We also show there is a language which cannot be accepted by any SCA in
  $o(n / \log n)$ time, even though it admits an $O(\log n)$-space streaming
  algorithm with $O(\log n)$ update time.
\end{abstract}


\section{Introduction}

The ongoing quest for lower bounds in complexity theory has been an arduous but
by no means unfruitful one.
Recent developments have brought to light a phenomenon dubbed \emph{hardness
magnification} \cite{oliveira18_hardness_focs, chen19_hardness_focs,
mckay19_weak_stoc, chen20_beyond_itcs, oliveira19_hardness_ccc,
cheraghchi20_one-tape_eccc}, giving several examples of natural problems for
which even slightly non-trivial lower bounds are as hard to prove as major
complexity class separations such as $\P \neq \NP$.
Among these, the preeminent example appears to be the \emph{minimum circuit
size problem}:

\begin{definition}[MCSP]
  For a Boolean function $f\colon \binalph^n \to \binalph$, let $\trta(f)$
  denote the truth table representation of $f$ (as a binary string in
  $\binalph^+$ of length $\abs{\trta(f)} = 2^n$).
  For $s\colon \N_+ \to \N_+$, the \emph{minimum circuit size problem}
  $\MCSP[s]$ is the problem where, given such a truth table $\trta(f)$, one must
  answer whether there is a Boolean circuit $C$ on inputs of length $n$ and size
  at most $s(n)$ that computes $f$, that is, $C(x) = f(x)$ for every input
  $x \in \binalph^n$.
\end{definition}

It is a well-known fact that there is a constant $K > 0$ such that, for any
function $f$ on $n$ variables as above, there is a circuit of size at most
$K \cdot 2^n / n$ that computes $f$; hence, $\MCSP[s]$ is only non-trivial for
$s(n) < K \cdot 2^n / n$.
Furthermore, $\MCSP[s] \in \NP$ for any constructible $s$ and, since every
circuit of size at most $s(n)$ can be described by a binary string of
$O(s(n) \log s(n))$ length, if $2^{O(s(n) \log s(n))} \subseteq \poly(2^n)$
(e.g., $s(n) \in O(n / \log n)$), by enumerating all possibilities we have
$\MCSP[s] \in \P$.
(Of course, such a bound is hardly useful since $s(n) \in O(n / \log n)$ implies
the circuit is degenerate and can only read a strict subset of its inputs.)
For large enough $s(n) < K \cdot 2^n / n$ (e.g., $s(n) \ge n$), it is unclear
whether $\MCSP[s]$ is $\NP$-complete (under polynomial-time many-one
reductions); see also \cite{kabanets00_circuit_stoc, murray17_np-hardness_toc}.
Still, we remark there has been some recent progress regarding
$\NP$-completeness under \emph{randomized} many-one reductions for \emph{certain
variants} of MCSP \cite{ilango20_np-hardness_ccc}.

\citeauthor{oliveira18_hardness_focs} \cite{oliveira18_hardness_focs} and
\citeauthor{oliveira19_hardness_ccc} \cite{oliveira19_hardness_ccc} recently
analyzed hardness magnification in the average-case as well as in the worst-case
approximation (i.e., gap) settings of MCSP for various (uniform and non-uniform)
computational models.
Meanwhile, \citeauthor{mckay19_weak_stoc} \cite{mckay19_weak_stoc} showed
similar results hold in the standard (i.e., exact or gapless) worst-case setting
and proved the following magnification result for (single-pass) \emph{streaming
algorithms} (see Definition~\ref{def_stream_alg}), which is a very restricted
uniform model; indeed, as mentioned in \cite{mckay19_weak_stoc}, even string
equality (i.e., the problem of recognizing $\{ ww \mid w \in \binalph^+ \}$)
cannot be solved by streaming algorithms (with limited space).

\begin{theorem}[\cite{mckay19_weak_stoc}] \label{thm_mmw19}
  Let $s\colon \N_+ \to \N_+$ be time constructible and $s(n) \ge n$.
  If there is no $\poly(s(n))$-space streaming algorithm with $\poly(s(n))$
  update time for (the search version of) $\MCSP[s]$, then $\P \neq \NP$.
\end{theorem}

In this paper, we present the following hardness magnification result for a
(uniform) computational model which is provably \emph{even more restricted} than
streaming algorithms: \emph{shrinking cellular automata} (SCAs).
Here, $\Block_b$ refers to a slightly modified presentation of $\MCSP[s]$ that
is only needed due to certain limitations of the model (see further discussion
as well as Section~\ref{sec_block_langs}).

\begin{restatable}{theorem}{restatethmscahm} \label{thm_sca_hm}
  For a certain $m \in \poly(s(n))$, if $\Block_b(\MCSP[s]) \not\in \SCA[n
  \cdot f(m)]$ for every $f \in \poly(m)$ and $b \in O(f)$, then $\P \neq \NP$.
\end{restatable}

Furthermore, we show every language accepted by a sublinear-time SCA can also be
accepted by a streaming algorithm of comparable complexity:

\begin{restatable}{theorem}{restatethmscatostreaming}%
  \label{thm_sca_to_streaming}%
  Let $t\colon \N_+ \to \N_+$ be computable by an $O(t)$-space random access
  machine (as in Definition~\ref{def_stream_alg}) in $O(t \log t)$ time.
  Then, if $L \in \SCA[t]$, there is an $O(t)$-space streaming algorithm for $L$
  with $O(t \log t)$ update and $O(t^2 \log t)$ reporting time.
\end{restatable}

Finally, we identify and prove \emph{two distinct limitations} of SCAs compared
to streaming algorithms (under sublinear-time constraints):
\begin{enumerate}
  \item They are insensitive to the length of long unary substrings in their
    input (Lemma~\ref{lem_pumping_SCA}), which means (standard versions of)
    fundamental problems such as parity, modulo, majority, and threshold cannot
    be solved in sublinear time (Proposition~\ref{prop_parity} and
    Corollary~\ref{cor_mod_maj_thr}).
  \item Only a limited amount of information can be transferred between cells
    which are far apart (in the sense of one-way communication complexity; see
    Lemma~\ref{lem_comm_compl_sca}).
\end{enumerate}
Both limitations are inherited from the underlying model of cellular automata.
The first can be avoided by presenting the input in a special format (the
previously mentioned $\Block_n$) that is efficiently verifiable by SCAs, which
we motivate and adopt as part of the model (see the discussion below).
The second is more dramatic and results in lower bounds even for languages
presented in this format:

\begin{restatable}{theorem}{restatethmlinearblocklang}%
  \label{thm_linear_block_lang}%
  There is a language $L_1$ for which $\Block_n(L_1) \not\in \SCA[o(N / \log
  N)]$ ($N$ being the instance length) can be accepted by an $O(\log N)$-space
  streaming algorithm with $\tilde{O}(\log N)$ update time.
\end{restatable}

From the above, it follows that any proof of $\P \neq \NP$ based on a lower
bound for solving $\MCSP[s]$ with streaming algorithms and
Theorem~\ref{thm_mmw19} must implicitly contain a proof of a lower bound for
solving $\MCSP[s]$ with SCAs.
From a more \enquote{optimistic} perspective (with an eventual proof of
$\P \neq \NP$ in mind), although not as widely studied as streaming algorithms,
SCAs are thus at least as good as a \enquote{target} for proving lower bounds
against and, in fact, should be an easier one if we are able to exploit
their aforementioned limitations.
Refer to Section~\ref{sec_conclusion} for further discussion on this, where we
take into account a recently proposed barrier \cite{chen20_beyond_itcs} to
existing techniques and which also applies to our proof of
Theorem~\ref{thm_linear_block_lang}.

From the perspective of cellular automata theory, our work furthers knowledge in
sublinear-time cellular automata models, a topic seemingly neglected by the
community at large (as pointed out in, e.g.,
\cite{modanese20_sublinear-time_dlt}).
Although this is certainly not the first result in which complexity-theoretical
results for cellular automata and their variants have consequences for classical
models (see, e.g., \cite{kutrib14_complexity_automata, poupet07_padding_csr}
for results in this sense), to the best of our knowledge said results address
only \emph{necessary} conditions for separating classical complexity classes.
Hence, our result is also novel in providing an implication in the other
direction, that is, a \emph{sufficient} condition for said separations based on
lower bounds for cellular automata models.

\ifconf
\else
\subsection{The Model}
\fi

(One-dimensional) cellular automata (CAs) are a parallel computational model
composed of identical \emph{cells} arranged in an array.
Each cell operates as a deterministic finite automaton (DFA) that is connected
with its left and right neighbors and operates according to the same local rule.
In classical CAs, the cell structure is immutable; \emph{shrinking} CAs relax
the model in that regard by allowing cells to spontaneously vanish (with their
contents being irrecoverably lost).
The array structure is conserved by reconnecting every cell with deleted
neighbors to the nearest non-deleted ones in either direction.

SCAs were introduced by \citeauthor{rosenfeld83_fast_is} in
\citeyear{rosenfeld83_fast_is} \cite{rosenfeld83_fast_is}, but it was not
until recent years that the model received greater attention by the CA community
\cite{kutrib15_shrinking_automata, modanese16_shrinking_automata}.
SCAs are a natural and robust model of parallel computation which, unlike
classical CAs, admit (non-trivial) sublinear-time computations.

We give a brief intuition as to how shrinking augments the classical CA model in
a significant way.
Intuitively speaking, any two cells in a CA can only communicate by signals,
which necessarily requires time proportional to the distance between them.
Assuming the entire input is relevant towards acceptance, this imposes a linear
lower bound on the time complexity of the CA.
In SCAs, however, this distance can be shortened as the computation evolves,
thus rendering acceptance in sublinear time possible.
As a matter of fact, the more cells are deleted, the faster distant cells can
communicate and the computation can evolve.
This results in a trade-off between space (i.e., cells containing information)
and time (i.e., amount of cells deleted).

\paragraph{Comparison with Related Models.}

Unlike other parallel models such as random access machines, SCAs are
\emph{incapable of random access} to their input.
In a similar sense, SCAs are constrained by the \emph{distance} between cells,
which is an aspect usually disregarded in circuits and related models except
perhaps for VLSI complexity \cite{thompson80_complexity_phd,
chazelle85_model_jacm}, for instance.
In contrast to VLSI circuits, however, in SCAs distance is a fluid aspect,
changing dynamically as the computation evolves.
Also of note is that SCAs are a \emph{local} computational model in a quite
literal sense of locality that is coupled with the above concept of distance
(instead of more abstract notions such as that from \cite{yao89_circuits_stoc},
for example).

These limitations hold not only for SCAs but also for standard CAs.
Nevertheless, SCAs are more powerful than other CA models capable of
sublinear-time computation such as ACAs \cite{modanese20_sublinear-time_dlt,
ibarra85_fast_tcs}, which are CAs with their acceptance behavior such that the
CA accepts if and only if all cells simultaneously accept.
This is because SCAs can \emph{efficiently aggregate results} computed in
parallel (by combining them using some efficiently computable function); in ACAs
any such form of aggregation is fairly limited as the underlying cell structure
is static.

\paragraph{Block Words.}

As mentioned above, there is an input format which allows us to circumvent the
first of the limitations of SCAs compared to streaming algorithms and which is
essential in order to obtain a more serious computational model.
In this format, the input is subdivided into \emph{blocks} of the same size and
which are separated by delimiters and numbered in ascending order from left to
right.
Words with this structure are dubbed \emph{block words} accordingly, and a
set of such words is a \emph{block language}.
There is a natural presentation of any (ordinary) word as a block word (by
mapping every symbol to its own block), which means there is a block language
version to any (ordinary) language.
(See Section~\ref{sec_block_langs}.)

The concept of block words seems to arise naturally in the context of
sublinear-time (both shrinking and standard) CAs
\cite{modanese20_sublinear-time_dlt, ibarra85_fast_tcs}.
The syntax of block words is very efficiently verifiable (more precisely, in
time linear in the block length) by a CA (without need of shrinking).
In addition, the translation of a language to its block version (and its
inverse) is a very simple map; one may frame it, for instance, as an $\AC^0$
reduction.
Hence, the difference between a language and its block version is solely in
presentation.

Block words coupled with CAs form a computational paradigm that appears to be
substantially diverse from linear- and real-time CA computation (see
\cite{modanese20_sublinear-time_dlt} for examples).
Often we shall describe operations on a block (rather than on a cell) level and,
by making use of block numbering, two blocks with distinct numbers may operate
differently even though their contents are the same; this would be impossible at
a cell level due to the locality of CA rules.
In combination with shrinking, certain block languages admit merging groups of
blocks in parallel; this gives rise to a form of reduction we call
\emph{blockwise reductions} and which we employ in a manner akin to downward
self-reducibility as in \cite{allender10_amplifying_jacm}.

An additional technicality which arises is that the number of cells in a block
is fixed at the start of the computation; this means a block cannot
\enquote{allocate extra space} (beyond a constant multiple of the block length).
This is the same limitation as that of linear bounded automata (LBAs) compared
to Turing machines with unbounded space, for example.
We cope with this limitation by increasing the block length in the problem
instances as needed, that is, by padding each block so that enough space is
available from the outset.\footnote{%
  An alternative solution is allowing the CA to \enquote{expand} by dynamically
  creating new cells between existing ones; however, this may result in a
  computational model which is dramatically more powerful than standard CAs
  \cite{modanese16_shrinking_automata,
    modanese19_complexity-theoretic_automata}.}
This is still in line with the considerations above; for instance, the resulting
language is still $\AC^0$ reducible to the original one (and vice-versa).

\ifconf
\else
\subsection{Techniques}
\fi

We give a broad overview of the proof ideas behind our results.

Theorem~\ref{thm_sca_hm} is a direct corollary of
Theorem~\ref{thm_sca_hm_explicit}, proven is Section~\ref{sec_sca_hm}.
The proof closely follows \cite{mckay19_weak_stoc} (see the discussion in
Section~\ref{sec_sca_hm} for a comparison) and, as mentioned above, bases on a
scheme similar to self-reducibility as in \cite{allender10_amplifying_jacm}.

The lower bounds in Section~\ref{sec_parity_etc} are established using
Lemma~\ref{lem_pumping_SCA}, which is a generic technical limitation of
sublinear-time models based on CAs (the first of the two aforementioned
limitations of SCAs with respect to streaming algorithms) and which we also show
to hold for SCAs.

One of the main technical highlights is the proof of
Theorem~\ref{thm_sca_to_streaming}, where we give a streaming algorithm to
simulate an SCA with limited space.
Our general approach bases on dynamic programming and is able to cope with the
unpredictability of when, which, or even how many cells are deleted during the
simulation.
The space efficiency is achieved by keeping track of only as much information as
needed as to determine the state of the SCA's decision cell step for step.

A second technical contribution is the application of \emph{one-way}
communication complexity to obtain lower bounds for SCAs, which yields
Theorem~\ref{thm_linear_block_lang}.
Essentially, we split the input in some position $i$ of our choice (which may
even be non-uniformly dependent on the input length) and have $A$ be given as
input the symbols preceding $i$ while $B$ is given the rest, where $A$ and $B$
are (non-uniform) algorithms with unbounded computational resources.
We show that, in this setting, $A$ can determine the state of the SCA's decision
cell with only $O(1)$ information from $B$ for every step of the SCA.
Thus, an SCA with time complexity $t$ for a language $L$ yields a protocol with
$O(t)$ one-way communication complexity for the above problem.
Applying this in the contrapositive, Theorem~\ref{thm_linear_block_lang} then
follows from the existence of a language $L_1$ (in some contexts referred to as
the indexing or memory access problem) that has nearly linear one-way
communication complexity despite admitting an efficient streaming algorithm.

\ifconf
\else
\subsection{Organization}
\fi

The rest of the paper is organized as follows:
Section~\ref{sec_preliminaries} presents the basic definitions.
In Section~\ref{sec_cap_lim_scas} we introduce block words and related concepts
and discuss the aforementioned limitations of sublinear-time SCAs.
Following that, in Section~\ref{sec_sca_streaming} we address the proof of
Theorem~\ref{thm_sca_to_streaming} and in Section~\ref{sec_sca_hm} that of
Theorem~\ref{thm_sca_hm}.
Finally, Section~\ref{sec_conclusion} concludes the paper.


\section{Preliminaries}
\label{sec_preliminaries}

We denote the set of integers by $\Z$, that of positive integers by $\N_+$, and
$\N_+ \cup \{ 0 \}$ by $\N_0$.
For $a,b \in \N_0$, $[a,b] = \{ x \in \N_0 \mid a \le x \le b \}$.
For sets $A$ and $B$, $B^A$ is the set of functions $A \to B$.

We assume the reader is familiar with cellular automata as well as with the
fundamentals of computational complexity theory (see, e.g., standard references
\cite{delorme99_cellular_book, goldreich08_computational_book,
arora09_computational_book}).
Words are indexed starting with index zero.
For a finite, non-empty set $\Sigma$, $\Sigma^\ast$ denotes the set of words
over $\Sigma$, and $\Sigma^+$ the set $\Sigma^\ast \setminus \{ \eps \}$.
For $w \in \Sigma^\ast$, we write $w(i)$ for the $i$-th symbol of $w$ (and, in
general, $w_i$ stands for another word altogether, \emph{not} the $i$-th symbol
of $w$).
For $a, b \in \N_0$, $w[a,b]$ is the subword $w(a) w(a+1) \cdots w(b-1) w(b)$ of
$w$ (where $w[a,b] = \eps$ for $a > b$).
$\abs{w}_a$ is the number of occurrences of $a \in \Sigma$ in $w$.
$\bin_n(x)$ stands for the binary representation of $x \in \N_0$, $x < 2^n$,
of length $n \in \N_+$ (padded with leading zeros).
$\poly(n)$ is the class of functions polynomial in $n \in \N_0$.
$\REG$ denotes the class of regular languages, and $\TISP[t,s]$ (resp.,
$\TIME[t]$) that of problems decidable by a Turing machine (with one tape and
one read-write head) in $O(t)$ time and $O(s)$ space (resp., unbounded space).
Without restriction, we assume the empty word $\eps$ is not a member of any of
the languages considered.

An \emph{$\omega$-word} is a map $\N_0 \to \Sigma$, and a \emph{$\domega$-word}
is a map $\Z \to \Sigma$.
We write $\Sigma^\omega = \Sigma^{\N_0}$ for the set of $\omega$-words over
$\Sigma$.
For $x \in \Sigma$, $x^\omega$ denotes the (unique) $\omega$-word with
$x^\omega(i) = x$ for every $i \in \N_0$.
To each $\domega$-word $w$ corresponds a unique pair $(w_-, w_+)$ of
$\omega$-words $w_-,w_+ \in \Sigma^\omega$ with $w_+(i) = w(i)$ for $i \ge 0$
and $w_-(i) = w(-i-1)$ for $i < 0$.
(Partial) $\omega$-word homomorphisms are extendable to (partial)
\emph{$\domega$-word homomorphisms} as follows:
Let $f\colon \Sigma^\omega \to \Sigma^\omega$ be an $\omega$-word homomorphism;
then there is a unique
$f_{\domega}\colon \Sigma^\Z \to \Sigma^\Z$ such that, for every
$w \in \Sigma^\Z$, $w' = f_{\domega}(w)$ is the $\domega$-word with
$w'_+ = f(w_+)$ and $w'_- = f(w_-)$.

For a circuit $C$, $\abs{C}$ denotes the \emph{size} of $C$, that is,
the total number of gates in $C$.
It is well-known that any Boolean circuit $C$ can be described by a binary
string of $O(\abs{C} \log \abs{C})$ length.

\begin{definition}[Streaming algorithm] \label{def_stream_alg}
  Let $s,u,r\colon \N_+ \to \N_+$ be functions.
  An \emph{$s$-space streaming algorithm} $A$ is a random access machine which,
on input $w$, works in $O(s(\abs{w}))$ space and, on every step, can either
perform an operation on a constant number of bits in memory or read the next
symbol of $w$.
  $A$ has $u$ \emph{update time} if, for every $w$, the number of operations it
performs between reading $w(i)$ and $w(i+1)$ is at most $u(\abs{w})$.
  $A$ has $r$ \emph{reporting time} if it performs at most $r(\abs{w})$
operations after having read $w(\abs{w}-1)$ (until it terminates).
\end{definition}

Our interest lies in $s$-space streaming algorithms that, for an input $w$, have
$\poly(s(\abs{w}))$ update and reporting time for sublinear $s$ (i.e.,
$s(\abs{w}) \in o(\abs{w})$).

\ifconf
\else
\subsection{Cellular Automata}
\fi

We consider only CAs with the standard neighborhood.
The symbols of an input $w$ are provided from left to right in the cells
$0$ to $|w| - 1$ and are surrounded by inactive cells, which conserve their
state during the entire computation (i.e., the CA is bounded).
Acceptance is signaled by cell zero (i.e., the leftmost input cell).

\begin{definition}[Cellular automaton] \label{def_CA}
  A \emph{cellular automaton} (\emph{CA}) $C$ is a tuple $(Q,\delta,\Sigma,q,A)$
  where:
  \begin{itemize*}[label={}, afterlabel={}]
    \item $Q$ is a non-empty and finite set of \emph{states};
    \item $\delta\colon Q^3 \to Q$ is the \emph{local transition function};
    \item $\Sigma \subsetneq Q$ is the \emph{input alphabet} of $C$;
    \item $q \in Q \setminus \Sigma$ is the \emph{inactive state}, that is,
      $\delta(q_1,q,q_2) = q$ for every $q_1,q_2 \in Q$; and
    \item $A \subseteq Q \setminus \{ q \}$ is the set of
      \emph{accepting states} of $C$.
  \end{itemize*}
  A cell which is not in the inactive state is said to be \emph{active}.
  The elements of $Q^\Z$ are the (\emph{global}) \emph{configurations} of $C$.
  $\delta$ induces the \emph{global transition function}
  $\Delta\colon Q^\Z \to Q^\Z$ of $C$ by
  $\Delta(c)(i) = \delta(c(i-1),c(i),c(i+1))$ for every cell $i \in \Z$ and
  configuration $c \in Q^\Z$.

  $C$ \emph{accepts} an input $w \in \Sigma^+$ if cell zero is eventually in an
  accepting state, that is, there is $t \in \N_0$ such that $(\Delta^t(c_0))(0)
  \in A$, where $c_0 = c_0(w)$ is the \emph{initial configuration} (for $w$):
  $c_0(i) = w(i)$ for $i \in [0,\abs{w}-1]$, and $c_0(i) = q$ otherwise.
  For a minimal such $t$, we say $C$ accepts $w$ with \emph{time complexity}
  $t$.
  $L(A) \subseteq \Sigma^+$ denotes the set of words accepted by $C$.
  For $t\colon \N_+ \to \N_0$, $\CA[t]$ is the class of languages accepted by
  CAs with time complexity $O(t(n))$, $n$ being the input length.
\end{definition}

For convenience, we extend $\Delta$ in the obvious manner (i.e., as a map
induced by $\delta$) so it is also defined for every (finite) word
$w \in Q^\ast$.
For $\abs{w} \le 2$, we set $\Delta(w) = \eps$; for longer words,
$\abs{\Delta(w)} = \abs{w} - 2$ holds.

Some remarks concerning the classes $\CA[t]$:
$\CA[\poly] = \TISP[\poly,n]$ (i.e., the class of polynomial-time LBAs), and
$\CA[t] = \CA[1] \subsetneq \REG$ for every sublinear $t$.
Furthermore, $\CA[t] \subseteq \TISP[t^2,n]$ (where $t^2(n) = (t(n))^2$) and
$\TISP[t,n] \subseteq \CA[t]$.

\begin{definition}[Shrinking CA] \label{def_SCA}
  A \emph{shrinking CA} (\emph{SCA}) $S$ is a CA with a \emph{delete state}
  $\otimes \in Q \setminus (\Sigma \cup \{ q \})$.
  The global transition function $\Delta_\Shr$ of $S$ is given by applying the
  standard CA global transition function $\Delta$ (as in
  Definition~\ref{def_CA}) followed by removing all cells in the state
  $\otimes$; that is, $\Delta_\Shr = \Phi \circ \Delta$, where $\Phi\colon Q^\Z
  \to Q^\Z$ is the (partial) $\domega$-word homomorphism resulting from the
  extension to $Q^\Z$ of the map $\phi\colon Q \to Q$ with $\phi(\otimes) =
  \eps$ and $\phi(x) = x$ for $x \in Q \setminus \{ \otimes \}$.
  For $t\colon \N_+ \to \N_0$, $\SCA[t]$ is the class of languages accepted by
  SCAs with time complexity $O(t(n))$, where $n$ denotes the input length.
\end{definition}

Note that $\Phi$ is only partial since, for instance, any $\domega$-word in
$\otimes^\omega \cdot \Sigma^\ast \cdot \otimes^\omega$ has no proper image
(as it is not mapped to a $\domega$-word).
Hence, $\Delta_\Shr$ is also only a partial function (on $Q^\Z$); nevertheless,
$\Phi$ is total on the set of $\domega$-words in which $\otimes$ occurs only
finitely often and, in particular, $\Delta_\Shr$ is total on the set of
configurations arising from initial configurations for finite input words (which
is the setting we are interested in).

The acceptance condition of SCAs is the same as in Definition~\ref{def_CA}
(i.e., acceptance is dictated by cell zero).
Unlike in standard CAs, the index of one same cell can differ from one
configuration to the next; that is, a cell index does not uniquely determine a
cell on its own (rather, only in conjunction with a time step).
This is a consequence of applying $\Phi$, which contracts the global
configuration towards cell zero.
More precisely, for a configuration $c \in Q^\Z$, the cell with index $i \ge 0$
in $\Delta(c)$ corresponds to that with index $i + d_i$ in $c$, where $d_i$ is
the number of cells with index $\le i$ in $c$ that were deleted in the
transition to $\Delta(c)$.
This also implies the cell with index zero in $\Delta(c)$ is the same as that in
$c$ with minimal positive index that was not deleted in the transition to
$\Delta(c)$; thus, in any time step, cell zero is the leftmost active cell
(unless all cells are inactive; in fact, cell zero is inactive if and only if
all other cells are inactive).
Granted, what indices a cell has is of little importance when one is interested
only in the configurations of an SCA and their evolution; nevertheless, they are
relevant when simulating an SCA with another machine model (as we do in
Sections~\ref{sec_linear_block_lang}~and~\ref{sec_sca_streaming}).

Naturally, $\CA[t] \subseteq \SCA[t]$ for every function $t$, and
$\SCA[\poly] = \CA[\poly]$.
For sublinear $t$, $\SCA[t]$ contains non-regular languages if, for instance,
$t \in \Omega(\log n)$ (see below); hence, the inclusion of $\CA[t]$ in
$\SCA[t]$ in strict.
In fact, this is the case even if we consider only regular languages.
One simple example is $L = \{ w \in \binalph^+ \mid w(0) = w(\abs{w}-1) \}$,
which is in $\SCA[1]$ and regular but not in $\CA[o(n)] = \CA[O(1)]$.
One obtains an SCA for $L$ by having all cells whose both neighbors are active
delete themselves in the first step; the two remaining cells then compare their
states, and cell zero accepts if and only if this comparison succeeds or if the
input has length $1$ (which it can notice immediately since it is only for such
words that it has two inactive neighbors).
Formally, the local transition function $\delta$ is such that, for $z_1,z_3 \in
\{ 0,1,q \}$ and $z_2 \in \binalph$, $\delta(z_1,z_2,z_3) = \otimes$ if both
$z_1$ and $z_3$ are in $\binalph$, $\delta(z_1,z_2,z_3) = z_2'$ if $z_1 = q$ or
$z_3 = q$, and $\delta(q,z_2',z_2') = \delta(q,z_2',q) = a$; in all other cases,
$\delta$ simply conserves the cell's state.
See Figure~\ref{fig_SCA_0ast} for an example.

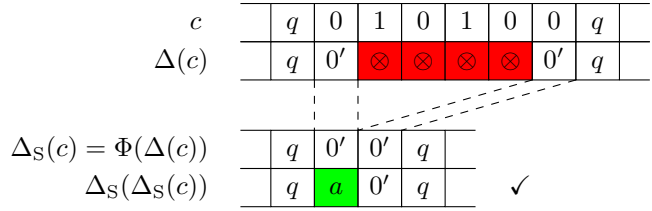
\begin{figure}[t]
  \centering
  \begin{tikzpicture}[every node/.style={block},
      block/.style={
        minimum width=width("$\otimes$") + 3mm,
        minimum height=height("$\otimes$") + 3mm,
        text height=1.5ex, text depth=.3ex, outer sep=0pt,
        draw, rectangle, node distance=0pt}]

    \setlength{\rowspc}{1.9em}
    \setlength{\lblspc}{2em}

    \node (a0)               {$q$};
    \node (a1) [right=of a0] {$0$};
    \node (a2) [right=of a1] {$1$};
    \node (a3) [right=of a2] {$0$};
    \node (a4) [right=of a3] {$1$};
    \node (a5) [right=of a4] {$0$};
    \node (a6) [right=of a5] {$0$};
    \node (a7) [right=of a6] {$q$};
    \draw (a0.north west) -- ++(-0.4cm,0)
          (a0.south west) -- ++(-0.4cm,0)
          (a7.north east) -- ++(0.4cm,0)
          (a7.south east) -- ++(0.4cm,0);

    \node (al) [draw=none, left=of a0, xshift=-\lblspc] {$c$};

    \node (b0) [below=of a0]          {$q$};
    \node (b1) [right=of b0]          {$0'$};
    \node (b2) [right=of b1,fill=red] {$\otimes$};
    \node (b3) [right=of b2,fill=red] {$\otimes$};
    \node (b4) [right=of b3,fill=red] {$\otimes$};
    \node (b5) [right=of b4,fill=red] {$\otimes$};
    \node (b6) [right=of b5]          {$0'$};
    \node (b7) [right=of b6]          {$q$};
    \draw (b0.north west) -- ++(-0.4cm,0)
          (b0.south west) -- ++(-0.4cm,0)
          (b7.north east) -- ++(0.4cm,0)
          (b7.south east) -- ++(0.4cm,0);

    \node (bl) [draw=none, left=of b0, xshift=-\lblspc] {$\Delta(c)$};

    \node (c0) [below=of b0,yshift=-\rowspc] {$q$};
    \node (c1) [right=of c0]                 {$0'$};
    \node (c2) [right=of c1]                 {$0'$};
    \node (c3) [right=of c2]                 {$q$};
    \draw (c0.north west) -- ++(-0.4cm,0)
          (c0.south west) -- ++(-0.4cm,0)
          (c3.north east) -- ++(0.4cm,0)
          (c3.south east) -- ++(0.4cm,0);

    \draw[dashed] (b1.south west) -- (c1.north west)
                  (b1.south east) -- (c1.north east)
                  (b6.south west) -- (c2.north west)
                  (b6.south east) -- (c2.north east);

    \node (cl) [draw=none, left=of c0, xshift=-\lblspc]
      {$\Delta_\Shr(c) = \Phi(\Delta(c))$};

    \node (d0) [below=of c0]            {$q$};
    \node (d1) [right=of d0,fill=green] {$a$};
    \node (d2) [right=of d1]            {$0'$};
    \node (d3) [right=of d2]            {$q$};
    \draw (d0.north west) -- ++(-0.4cm,0)
          (d0.south west) -- ++(-0.4cm,0)
          (d3.north east) -- ++(0.4cm,0)
          (d3.south east) -- ++(0.4cm,0);

    \node (dl) [draw=none, left=of d0, xshift=-\lblspc]
      {$\Delta_\Shr(\Delta_\Shr(c))$};

    \node (cm) [draw=none, right=of d3, xshift=\lblspc] {\checkmark};
  \end{tikzpicture}
  \caption{Computation of an SCA that recognizes
    $L = \{ w \in \binalph^+ \mid w(0) = w(\abs{w}-1) \}$ in $O(1)$ time.
    Here, the input word is $010100 \in L$.}
  \label{fig_SCA_0ast}
\end{figure}

Using a textbook technique to simulate a (bounded) CA with an LBA (simply
skipping deleted cells), we have:

\begin{proposition}
  For every function $t\colon \N_+ \to \N_+$ computable by an LBA in
  $O(n \cdot t(n))$ time, $\SCA[t] \subseteq \TISP[n \cdot t(n), n]$.
\end{proposition}

The inclusion is actually proper (see Corollary~\ref{cor_SCA_neq_TISP}).
Using the well-known result that $\TIME[o(n \log n)] = \REG$
\cite{kobayashi85_structure_tcs}, it follows that at least a logarithmic time
bound is needed for SCAs to recognize languages which are not regular:

\begin{corollary}
  $\SCA[o(\log)] \subsetneq \REG$.
\end{corollary}

This bound is tight:
It is relatively easy to show that any language accepted by ACAs in $t(n)$ time
can also be accepted by an SCA in $t(n) + O(1)$ time.
Since there is a non-regular language recognizable by ACAs
\cite{ibarra85_fast_tcs} in $O(\log n)$ time, the same language is recognizable
by an SCA in $O(\log n)$ time.

For any finite, non-empty set $\Sigma$, we say a function
$f\colon \Sigma^+ \to \Sigma^+$ is \emph{computable in place} by an (S)CA if
there is an (S)CA $S$ which, given $x \in \Sigma^+$ as input (surrounded by
inactive cells), produces $f(x)$.
Additionally, $g\colon \N_+ \to \N_+$ is \emph{constructible in place} by an
(S)CA if $g(n) \le 2^n$ and there is an (S)CA $S$ which, given $n \in \N_0$ in
unary, produces $\bin_n(g(n) - 1)$ (i.e., $g(n)-1$ in binary).
Note the set of functions computable or constructible in place by an (S)CA in at
most $t(n)$ time, where $n$ is the input length and $t\colon \N_+ \to \N_+$ is
some function, includes (but is not limited to) all functions computable by an
LBA in at most $t(n)$ time.


\section{Capabilities and Limitations of Sublinear-Time SCAs}
\label{sec_cap_lim_scas}

\subsection{Block Languages}
\label{sec_block_langs}

Let $\Sigma$ be a finite, non-empty set.
For $\Sigma_\eps = \Sigma \cup \{ \eps \}$ and $x,y \in \Sigma^+$,
$\binom{x}{y}$ denotes the (unique) word in $(\Sigma_\eps \times \Sigma_\eps)^+$
of length $\max\{ \abs{x}, \abs{y} \}$ for which
$\binom{x}{y}(i) = (x(i),y(i))$, where $x(i) = y(j) = \eps$ for $i \ge \abs{x}$
and $j \ge \abs{y}$.

\begin{definition}[Block word] \label{def_block_word}
  Let $n, m, b \in \N_+$ be such that $b \ge n$ and $m \le 2^n$.
  A word $w$ is said to be an $(n,m,b)$-\emph{block word} (over $\Sigma$) if it
  is of the form $w = w_0 \# w_1 \# \cdots \# w_{m-1}$ and $w_i =
  \binom{\bin_n(x_i)}{y_i}$, where $x_0 \ge 0$, $x_{i+1} = x_i + 1$ for every
  $i$, $x_{m-1} < 2^n$, and $y_i \in \Sigma^b$.
  In this context, $w_i$ is the \emph{$i$-th block} of $w$.
\end{definition}

Hence, every $(n,m,b)$-block word $w$ has $m$ many blocks of length $b$,
and its total length is $\abs{w} = (b + 1) \cdot m - 1 \in \Theta(bm)$.
For example,
\[
  w = \binom{01}{0100} \# \binom{10}{1100} \# \binom{11}{1000}
\]
is a $(2,3,4)$-block word with $x_0 = 1$, $y_0 = 0100$, $y_1 = 1100$, and
$y_2 = 1000$.
$n$ is implicitly encoded by the entries in the upper track (i.e., the $x_i$)
and we shall see $m$ and $b$ as parameters depending on $n$ (see
Definition~\ref{def_block_lang} below), so the structure of each block can be
verified locally (i.e., by inspecting the immediate neighborhood of every
block).
Note the block numbering starts with an arbitrary $x_0$; this is intended so
that, for $m' < m$, an $(n,m,b)$-block word admits $(n,m',b)$-block words as
infixes (which would not be the case if we required, say, $x_0 = 0$).

When referring to block words, we use $N$ for the block word length $\abs{w}$
and reserve $n$ for indexing block words of different block length, overall
length, or total number of blocks (or any combinations thereof).
With $m$ and $b$ as parameters depending on $n$, we obtain sets of block words:

\begin{definition}[Block language] \label{def_block_lang}
  Let $m,b\colon \N_+ \to \N_+$ be non-decreasing and constructible in place by
  a CA in $O(m(n) + b(n))$ time.
  Furthermore, let $b(n) \ge n$ and $m(n) \le 2^n$.
  Then, $\B^m_b$ denotes the set of all $(n,m(n),b(n))$-block words for
  $n \in \N_+$, and every subset $L \subseteq \B^m_b$ is an
  (\emph{$(n,m,b)$-})\emph{block language} (over $\Sigma$).
\end{definition}

An SCA can \emph{verify} its input is a valid block word in $O(b(n))$ time, that
is, locally check that the structure and contents of the blocks are consistent
(i.e., as in Definition~\ref{def_block_word}).
This can be realized using standard CA techniques without need of shrinking (see
\cite{modanese20_sublinear-time_dlt, ibarra85_fast_tcs} for constructions).
Recall Definition~\ref{def_SCA} does not require an SCA $S$ to explicitly reject
inputs not in $L(S)$, that is, the time complexity of $S$ on an input $w$ is
only defined for $w \in L(S)$.
As a result, when $L(S)$ is a block language, the time spent verifying that $w$
is a block word is only relevant if $w \in L(S)$ and, in particular, if $w$ is a
(valid) block word.
Provided the state of every cell in $S$ eventually impacts its decision to
accept (which is the case for all constructions we describe), it suffices to
have a cell mark itself with an error flag whenever a violation in $w$ is
detected (even if other cells continue their operation as normal); since every
cell is relevant towards acceptance, this eventually prevents $S$ from accepting
(and, since $w \not\in L(S)$, it is irrelevant how long it takes for this to
occur).
Thus, for the rest of this paper, when describing an SCA for a block language,
we implicitly require that the SCA checks its input is a valid block word
beforehand.

As stated in the introduction, our interest in block words is as a special input
format.
There is a natural bijection between any language and a block version of it,
namely by mapping each word $z$ to a block word $w$ in which each block $w_i$
contains a symbol $z(i)$ of $z$ (padded up to the block length $b$) and the
blocks are numbered from $0$ to $\abs{z} - 1$:

\begin{definition}[Block version of a language] \label{def_block_ver}
  Let $L \subseteq \Sigma^+$ be a language and $b$ as in
  Definition~\ref{def_block_lang}.
  The \emph{block version} $\Block_b(L)$ of $L$ (with blocks of length $b$) is
  the block language for which, for every $z \in \Sigma^+$, $z \in L$ holds if
  and only if we have $w \in \Block_b(L)$ where $w$ is the $(n,m,b(n))$-block
  word (as in Definition~\ref{def_block_word}) with $m = \abs{z}$, $n =
  \ceil{\log m}$, $x_0 = 0$, and $y_i = z(i)0^{b(n)-1}$ for every $i \in
  [0,m-1]$.
\end{definition}

Note that, for any such language $L$, $\Block_b(L) \not\in \REG$ for any $b$
(since $b(n) \ge n$ is not constant); hence, $\Block_b(L) \in \SCA[t]$ only for
$t \in \Omega(\log n)$ (and constructible $b$).
For $b(n) = n$, $\Block_n(L)$ is the block version with minimal padding.

For any two finite, non-empty sets $\Sigma_1$ and $\Sigma_2$, say a function
$f\colon \Sigma_1^+ \to \Sigma_2^+$ is \emph{non-stretching} if
$\abs{f(x)} \le \abs{x}$ for every $x \in \Sigma_1^+$.
We now define $k$-blockwise maps, which are maps that operate on block words by
grouping $k(n)$ many blocks together and mapping each such group (in a
non-stretching manner) to a single block of length at most
$(b(n) + 1) \cdot k(n) - 1$.

\begin{definition}[Blockwise map] \label{def_bw_map}
  Let $k\colon \N_+ \to \N_+$, $k(n) \ge 2$, be a non-decreasing function and
  constructible in place by a CA in $O(k(n))$ time.
  A map $g\colon \B^{km}_b \to \B^m_b$ is a \emph{$k$-blockwise map} if there
  is a non-stretching $g'\colon \B^k_b \to \Sigma^+$ such that, for every $w
  \in \B^{km}_b$ (as in Definition~\ref{def_block_word}) and $w_i' = w_{ik} \#
  \cdots \# w_{(i+1)k - 1}$:
  \[
    g(w) = \binom{\bin_n(x_0)}{g'(w_0')} \# \cdots \#
    \binom{\bin_n(x_{m-1})}{g'(w_{m-1}')}.
  \]
\end{definition}

Using blockwise maps, we obtain a very natural form of reduction operating on
block words and which is highly compatible with sublinear-time SCAs as a
computational model.
The reduction divides an $(n,km,b)$-block word in $m$ many groups of $k$ many
contiguous blocks and, as a $k$-blockwise map, maps each such group to a single
block (of length $b$):

\begin{definition}[Blockwise reducible] \label{def_bw_red}
  For block languages $L$ and $L'$, $L$ is (\emph{$k$-})\emph{blockwise
    reducible} to $L'$ if there is a computable $k$-blockwise map $g\colon
  \B^{km}_b \to \B^m_b$ such that, for every $w \in \B^{km}_b$, we have $w \in
  L$ if and only if $g(w) \in L'$.
\end{definition}

Since every application of the reduction reduces the instance length by a factor
of approximately $k$, logarithmically many applications suffice to produce a
trivial instance (i.e., an instance consisting of a single block).
This gives us the following computational paradigm of chaining blockwise
reductions together:

\begin{lemma} \label{lem_bw_reduction_chain}
  Let $k,r \colon \N_+ \to \N_0$ be functions, and let $L \subseteq \B^{k^r}_b$
  be such that there is a series $L = L_0, L_1, \dots, L_{r(n)}$ of languages
  with $L_i \subseteq \B^{k^{r-i}}_b$ and such that $L_i$ is $k(n)$-blockwise
  reducible to $L_{i+1}$ via the (same) blockwise reduction $g$.
  Furthermore, let $g'$ be as in Definition~\ref{def_bw_map}, and let
  $t_{g'}\colon \N_+ \to \N_+$ be non-decreasing and such that, for every
  $w' \in \B^r_b$, $g'(w')$ is computable in place by an SCA in
  $O(t_{g'}(\abs{w'}))$ time.
  Finally, let $L_{r(n)} \in \SCA[t]$ for some function $t\colon \N_+ \to \N_+$.
  Then, $L \in \SCA[r(n) \cdot t_{g'}(O(k(n) \cdot b(n))) + O(b(n)) + t(b(n))]$.
\end{lemma}

\begin{proof}
  We consider the SCA $S$ which, given $w \in \B^{k^r}_b$, repeatedly applies
  the reduction $g$, where each application of $g$ is computed by applying $g'$
  on each group of relevant blocks (i.e., the $w_i'$ from
  Definition~\ref{def_bw_map}) in parallel.

  One detail to note is that this results in the same procedure $P$ being
  applied to different groups of blocks in parallel, but it may be so that $P$
  requires more time for one group of blocks than for the other.
  Thus, we allow the entire process to be carried out asynchronously but require
  that, for each group of blocks, the respective results be present before each
  execution of $P$ is started.
  (One way of realizing this, for instance, is having the first block in the
  group send a signal across the whole group to ensure all inputs are available
  and, when it arrives at the last block in the group, another signal is sent
  to trigger the start of $P$.)

  Using that $t_{g'}$ is non-decreasing and that $g'$ is non-stretching, the
  time needed for each execution of $P$ is
  $t_{g'}(\abs{w_i'}) \in t_{g'}(O(k(n) \cdot b(n)))$
  (which is not impacted by the considerations above) and, since there are
  $r(n)$ reductions in total, we have $r(n) \cdot t_{g'}(O(k(n) \cdot b(n)))$
  time in total.
  Once a single block is left, the cells in this block synchronize themselves
  and then behave as in the SCA $S'$ for $L_{r(n)}$ guaranteed by the
  assumption; using a standard synchronization algorithm, this requires
  $O(b(n))$ for the synchronization, plus $t(b(n))$ time for emulating $S'$.
\end{proof}

\subsection{Block Languages and Parallel Computation}
\label{sec_parity_etc}

In this section, we prove the first limitation of SCAs discussed in the
introduction (Lemma~\ref{lem_pumping_SCA}) and which renders them unable of
accepting the languages $\PAR$, $\MOD_q$, $\MAJ$, and $\THR_k$ (defined next) in
sublinear time.
Nevertheless, as is shown in Proposition~\ref{prop_sca_par_etc}, the
\emph{block versions} of these languages can be accepted quite efficiently.
This motivates the block word presentation for inputs; that is, this first
limitation concerns only the \emph{presentation} of instances (and, hence, is
not a \emph{computational} limitation of SCAs).

Let $q > 2$ and let $k\colon \N_+ \to \N_+$ be constructible in place by a CA
in at most $t_k(n)$ time for some $t_k\colon \N_+ \to \N_+$.
Additionally, let $\PAR$ (resp., $\MOD_q$; resp., $\MAJ$; resp., $\THR_k$) be
the language consisting of every word $w \in \binalph^+$ for which $\abs{w}_1$
is even (resp., $\abs{w}_1 = 0 \pmod q$; resp., $\abs{w}_1 \ge \abs{w}_0$;
resp., $\abs{w}_1 \ge k(\abs{w})$).

The following is a simple limitation of sublinear-time CA models such as ACAs
(see also \cite{sommerhalder83_parallel_ai}) which we show also to hold for
SCAs.

\begin{lemma} \label{lem_pumping_SCA}
  Let $S$ be an SCA with input alphabet $\Sigma$, and let $x \in \Sigma$ be such
  that there is a minimal $t \in \N_+$ for which $\Delta_\Shr^t(y) = \eps$,
  where $y = x^{2t+1}$ (i.e., the symbol $x$ concatenated $2t+1$ times with
  itself).
  Then, for every $z_1,z_2 \in \Sigma^+$, $w = z_1 y z_2 \in L(S)$ holds if and
  only if for every $i \in \N_0$ we have $w_i = z_1 y x^i z_2 \in L(S)$.
\end{lemma}

\begin{proof}
  Given $w$ and $i$ as above, we show $w_i \in L(S)$; the converse is trivial.
  Since $w$ and $w_i$ both have $z_1y$ as prefix and
  $\Delta_\Shr^{t'}(y) \neq \eps$ for $t' < t$, if $S$ accepts $w$ in $t'$
  steps, then it also accepts $w_i$ in $t'$ steps.
  Thus, assume $S$ accepts $w$ in $t' \ge t$ steps, in which case it suffices to
  show $\Delta_\Shr^t(w) = \Delta_\Shr^t(w_i)$.
  To this end, let $\alpha_j$ for $j \in [0,t]$ be such that $\alpha_0 = x$ and
  $\alpha_{j+1} = \delta(\alpha_j,\alpha_j,\alpha_j)$.
  Hence, $\Delta(\alpha_j^{k+2}) = \alpha_{j+1}^k$ holds for every $k \in \N_+$
  (and $j < t$) and, by an inductive argument as well as by the assumption on
  $y$ (i.e., $\alpha_t = \otimes$),
  $\Delta_\Shr^t(yx^i) = \Delta_\Shr^t(\alpha_0^{2t+i+1}) = \eps$.
  Using this along with $\abs{y} \ge t$ and $y \in \{ x \}^+$, we have
  $\Delta_\Shr^t(q^tz_1yx^i) = \Delta_\Shr^t(q^tz_1y)$ and
  $\Delta_\Shr^t(yx^iz_2q^t) = \Delta_\Shr^t(x^iyz_2q^t) =
  \Delta_\Shr^t(yz_2q^t)$;
  hence, $\Delta_\Shr^t(w) = \Delta_\Shr^t(w_i)$ follows.
\end{proof}

An implication of Lemma~\ref{lem_pumping_SCA} is that every unary language
$U \in \SCA[o(n)]$ is either finite or cofinite.
As $\PAR \cap \{ 1 \}^+$ is neither finite nor cofinite, we can prove:

\begin{proposition} \label{prop_parity}
  $\PAR \not\in \SCA[o(n)]$ (where $n$ is the input length).
\end{proposition}

\begin{proof}
  Let $S$ be an SCA with $L(S) = \PAR$.
  We show $S$ must have $\Omega(n)$ time complexity on inputs from the infinite
  set $U = \{ 1^{2m} \mid m \in \N_+ \} \subset \PAR$.
  If $\Delta_\Shr^t(1^{2t+1}) = \eps$ for some $t \in \N_0$, then, by
  Lemma~\ref{lem_pumping_SCA}, $L(S) \cap \{ 1 \}^+$ is either finite or
  cofinite, which contradicts $L(S) = \PAR$.
  Hence, $\Delta_\Shr^t(1^{2t+1}) \neq \eps$ for every $t \in \N_0$.
  In this case, the trace of cell zero on input $w = 1 1^{2t+1} 1$ in the first
  $t$ steps is the same as that on input $w' = 1 1^{2t+1} 11$.
  Since $w \in \PAR$ if and only if $w' \not\in \PAR$, it follows that $S$ has
  $\Omega(t) = \Omega(n)$ time complexity on $U$.
\end{proof}

\begin{corollary} \label{cor_SCA_neq_TISP}
  $\REG \not\subseteq \SCA[o(n)]$.
\end{corollary}

The argument above generalizes to $\MOD_q$, $\MAJ$, and $\THR_k$ with $k \in
\omega(1)$.
For $\MOD_q$, consider $U = \{ 1^{qm} \mid m \in \N_+ \}$.
For $\MAJ$ and $\THR_k$, set $U = \{ 0^m 1^m \mid m \in \N_+ \}$ and
$U = \{ 0^{n-k(n)} 1^{k(n)} \mid n \in \N_+ \}$, respectively; in this case, $U$
is not unary, but the argument easily extends to the unary suffixes of the words
in $U$.

\begin{corollary} \label{cor_mod_maj_thr}
  $\MOD_q, \MAJ \not\in \SCA[o(n)]$.
  Also, $\THR_k \in \SCA[o(n)]$ if and only if $k \in O(1)$.
\end{corollary}

The \emph{block versions} of these languages, however, are not subject to the
limitation above:

\begin{proposition} \label{prop_sca_par_etc}
  For $L \in \{ \PAR, \MOD_q, \MAJ \}$, $\Block_n(L) \in \SCA[(\log N)^2]$,
  where $N = N(n)$ is the input length.
  Also, $\Block_n(\THR_k) \in \SCA[(\log N)^2 + t_k(n)]$.
\end{proposition}

\begin{proof}
  Given $L \in \{ \PAR, \MOD_q, \MAJ, \THR_k \}$, we construct an SCA $S$ for
  $L' = \Block_n(L)$ with the purported time complexity.
  Let $w \in \B^m_n$ be an input of $S$.
  For simplicity, we assume that, for every such $w$, $m = m(n) = 2^n$ is a
  power of two; the argument extends to the general case in a simple manner.
  Hence, we have $N = \abs{w} = n \cdot m$ and $n = \log m \in \Theta(\log N)$.

  Let $L_0 \subset \B^m_n$ be the language containing every such block word
  $w \in \B^m_n$ for which, for $y_i$ as in Definition~\ref{def_block_word} and
  $y = \sum_{i=0}^{m-1} y_i$, we have  $f_L(y) = f_{L,n}(y) = 0$, where
  $f_\PAR(y) = y \bmod 2$, $f_{\MOD_q}(y) = y \bmod q$, $f_\MAJ(y) = 0$ if and
  only if $y \ge 2^{n-1}$,  and $f_{\THR_k}(y) = 0$ if and only if $y \ge k(n)$.
  Thus, (under the previous assumption) we have $L_0 = L'$ (and, in the general
  case, $L_0 = L' \cap \B^{2^n}_n$).

  Then, $L_0$ is $2$-blockwise reducible to a language
  $L_1 \subseteq \B^{m/2}_n$ by mapping every $(n,2,n)$-block word of the form
  $\binom{\bin_n(2x)}{y_{2x}}\#\binom{\bin_n(2x+1)}{y_{2x+1}}$ with
  $x \in [0,2^{n-1} - 1]$ to $\binom{\bin_n(x)}{y_{2x} + y_{2x+1}}$.
  To do so, it suffices to compute $\bin_n(x)$ from $\bin_n(2x)$ and add the
  $y_{2x}$ and $y_{2x+1}$ values in the lower track; using basic CA arithmetic
  and cell communication techniques, this is realizable in $O(n)$ time.
  Repeating this procedure, we obtain a chain of languages $L_0, \dots, L_n$
  such that $L_i$ is $2$-blockwise reducible to $L_{i+1}$ in $O(n)$ time.
  By Lemma~\ref{lem_bw_reduction_chain}, $L' \in \SCA[n^2 + t(n)]$
  follows, where $t\colon \N_+ \to \N_0$ is such that $L_n \in \SCA[t]$.
  For $L \in \{ \PAR, \MOD_q, \MAJ \}$, checking the above condition on $f_L(y)$
  can be done in $t(n) \in O(n)$ time; as for $L = \THR_k$, we must also compute
  $k$, so we have $t(n) \in O(n + t_k(n))$.

  The general case follows from adapting the above reductions so that words with
  an odd number of blocks are also accounted for (e.g., by ignoring the last
  block of $w$ and applying the reduction on the first $m-1$ blocks).
\end{proof}

\subsection{An Optimal SCA Lower Bound for a Block Language}
\label{sec_linear_block_lang}

Corollary~\ref{cor_SCA_neq_TISP} already states SCAs are strictly less capable
than streaming algorithms.
However, the argument bases exclusively on long unary subwords in the input
(i.e., Lemma~\ref{lem_pumping_SCA}) and, therefore, does not apply to block
languages.
Hence Theorem~\ref{thm_linear_block_lang}, which shows SCAs are more limited
than streaming algorithms \emph{even considering only block languages}:

\restatethmlinearblocklang*

Let $L_1$ be the language of words $w \in \binalph^+$ such that $\abs{w} = 2^n$
is a power of two and, for $i = w(0) w(1) \cdots w(n-1)$ (seen as an $n$-bit
binary integer), $w(i) = 1$.
It is not hard to show that its block version $\Block_n(L_1)$ can be accepted
by an $O(\log m)$-space streaming algorithm with $\tilde{O}(\log m)$ update
time.

The $O(N/\log N)$ upper bound for $\Block_n(L_1)$ is optimal since there is an
$O(N / \log N)$ time SCA for it:
Shrink every block to its respective bit (i.e., the $y_i$ from
Definition~\ref{def_block_word}), reducing the input to a word $w'$ of
$O(N / \log N)$ length; while doing so, mark the bit corresponding to the $n$-th
block.
Then shift the contents of the first $n$ bits as a counter that decrements
itself every new cell it visits and, when it reaches zero, signals acceptance if
the cell it is currently at contains a $1$.
Using counter techniques as in \cite{stratmann02_leader_fgcs,
vollmar77_modified_ac}, this requires $O(\abs{w'})$ time.

The proof of Theorem~\ref{thm_linear_block_lang} bases on communication
complexity.
The basic setting is a game with two players $A$ and $B$ (both with unlimited
computational resources) which receive inputs $w_A$ and $w_B$, respectively, and
must produce an answer to the problem at hand while exchanging a limited amount
of bits.
We are interested in the case where the concatenation $w = w_A w_B$ of the
inputs of $A$ and $B$ is an input to an SCA and $A$ must output whether the SCA
accepts $w$.
More importantly, we analyze the case where \emph{only $B$ is allowed to send
messages}, that is, the case of \emph{one-way} communication.%
\footnote{One-way communication complexity can also been defined as the maximum
over \emph{both} communication directions (i.e., $B$ to $A$ and $A$ to $B$; see
\cite{durr04_cellular_tcs} for an example in the setting of CAs).
  Since our goal is to prove a \emph{lower bound} on communication complexity,
it suffices to consider a single (arbitrary) direction (in this case $B$ to
$A$).}

\begin{definition}[One-way communication complexity] \label{def_comm_compl}
  Let $m,f\colon \N_+ \to \N_+$ be functions with $0 < m(N) \le N$.
  A language $L \subseteq \Sigma^+$ is said to have \emph{($m$-)one-way
    communication complexity $f$} if there are families of algorithms (with
  unlimited computational resources) $(A_N)_{N \in \N_+}$ and $(B_N)_{N \in
    \N_+}$ such that the following holds for every $w \in \Sigma^\ast$ of length
  $\abs{w} = N$, where $w_A = w[0,m(N)-1]$ and $w_B = w[m(N),N-1]$:
  \begin{enumerate}
    \item $\abs{B_N(w_B)} \le f(N)$; and
    \item $A_N(w_A, B(w_B)) = 1$ (i.e., accept) if and only if $w \in L$.
  \end{enumerate}
  $\owComm^m(L)$ indicates the (pointwise) minimum over all such functions $f$.
\end{definition}

Note that $A_N$ and $B_N$ are nonuniform, so the length $N$ of the (complete)
input $w$ is known implicitly by both algorithms.

\begin{lemma} \label{lem_comm_compl_sca}
  For any computable $t\colon \N_+ \to \N_+$ and $m$ as in
  Definition~\ref{def_comm_compl}, if $L \in \SCA[t]$, then
  $\owComm^m(L)(N) \in O(t(N))$.
\end{lemma}

The proof idea is to have $A$ and $B$ simulate the SCA for $L$ simultaneously,
with $A$ maintaining the first half $c_A$ of the SCA configuration and $B$ the
second half $c_B$.
(Hence, $A$ is aware of the leftmost active state in the SCA and can detect
whether the SCA accepts or not.)
The main difficulty is guaranteeing that $A$ and $B$ can determine the states of
the cells on the right (resp., left) end of $c_A$ (resp., $c_B$) despite the
respective local configurations \enquote{overstepping the boundary} between
$c_A$ and $c_B$.
Hence, for each step in the simulation, $B$ communicates the states of the two
leftmost cells in $c_B$; with this, $A$ can compute the states of all cells of
$c_A$ in the next configuration as well as that of the leftmost cell $\alpha$ of
$c_B$, which is added to $c_A$.
(See Figure~\ref{fig_sca_comm_compl_lem} for an illustration.)
This last technicality is needed due to one-way communication, which renders it
impossible for $B$ to determine the next state of $\alpha$ (since its left
neighbor is in $c_A$ and $B$ cannot receive messages from $A$).
As the simulation requires at most $t(N)$ steps and $B$ sends $O(1)$ information
at each step, this yields the purported $O(t(N))$ upper bound.

The attentive reader may have noticed this discussion does not address the fact
that the SCA may shrink; indeed, we shall also prove that shrinking does not
interfere with this strategy.

\begin{proof}
  Let $S$ be an SCA for $L$ with time complexity $O(t)$.
  Furthermore, let $Q$ be the state set of $S$ and $q \in Q$ its inactive state.
  We construct algorithms $A_N$ and $B_N$ as in Definition~\ref{def_comm_compl}
  and such that $\abs{B_N(w_B)} \le 2 \log(\abs{Q}) \cdot t(N)$.

  Fix $N \in \N_+$ and an input $w \in \Sigma^N$.
  For $w_B^0 = w_B q^{2t(N) + 2}$ and $w_B^{i+1} = \Delta_\Shr(w_B^i)$ for
  $i \in \N_0$, $B_N$ computes and outputs the concatenation
  \[
    B_N(w_B)
    = w_B^0(0) w_B^0(1) w_B^1(0) w_B^1(1) \cdots w_B^{t(N)}(0) w_B^{t(N)}(1).
  \]
  Similarly, let $w_A^0 = q^{2t(N) + 2} w_A$ and
  $w_A^{i+1} = \Delta_\Shr(w_A^i w_B^i(0) w_B^i(1))$ for $i \in \N_0$.
  $A$ computes $t(N)$ and $w_A^i$ for $i \in [0,t(N)]$ and accepts if there is
  any $j$ such that $w_A^i(j)$ is an accept state of $S$ and $w_A^i(j') = q$ for
  all $j' < j$; otherwise, $A$ rejects.

  To prove the correctness of $A$, we show by induction on $i \in \N_0$:
  $w_A^i w_B^i = \Delta_\Shr^i(q^{2t(n) + 2} w q^{2t(n) + 2})$.
  Hence, the $w_A^i(j)$ of above corresponds to the state of cell zero in
  step $i$ of $S$, and it follows that $A$ accepts if and only if $S$ does.
  The induction basis is trivial.
  For the induction step, let $w' = \Delta_\Shr(w_A^i w_B^i)$.
  Using the induction hypothesis, it suffices to prove
  $w_A^{i+1} w_B^{i+1} = w'$.
  Note first that, due to the definition of $w_A^{i+1}$ and $w_B^{i+1}$, we
  have $w' = \Delta_\Shr(w_A^i) \alpha \beta \Delta_\Shr(w_B^i)$, where
  $\alpha, \beta \in Q \cup \{ \eps \}$.
  Let $\alpha_1 = w_A^i(\abs{w_A^i} - 2)$, $\alpha_2 = w_A^i(\abs{w_A^i} - 1)$,
  and $\alpha_3 = w_B^i(0)$ and notice
  $\alpha = \delta(\alpha_1, \alpha_2, \alpha_3)$; the same is true for
  $\beta$ and $\beta_1 = \alpha_2$, $\beta_2 = \alpha_3$, and
  $\beta_3 = w_B^i(1)$.
  Hence, we have $w_A^{i+1} = \Delta_\Shr(w_A^i) \alpha \beta$, and the claim
  follows.
\end{proof}

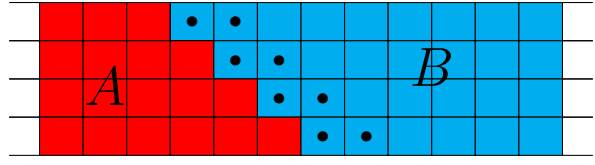
\begin{figure}[t]
  \centering
  \begin{tikzpicture}[every node/.style={block},
      block/.style={
        minimum width=width("$\otimes$") + 3mm,
        minimum height=height("$\otimes$") + 3mm,
        text height=1.5ex, text depth=.3ex, outer sep=0pt,
        draw, rectangle, node distance=0pt}]

    \setlength{\rowspc}{2em}
    \setlength{\lblspc}{2em}

    \tikzstyle{A}=[fill=red]
    \tikzstyle{B}=[fill=cyan]

    \node[A] (a0)               {};
    \node[A] (a1) [right=of a0] {};
    \node[A] (a2) [right=of a1] {};
    \node[B] (a3) [right=of a2] {$\bullet$};
    \node[B] (a4) [right=of a3] {$\bullet$};
    \node[B] (a5) [right=of a4] {};
    \node[B] (a6) [right=of a5] {};
    \node[B] (a7) [right=of a6] {};
    \node[B] (a8) [right=of a7] {};
    \node[B] (a9) [right=of a8] {};
    \node[B] (a10) [right=of a9] {};
    \node[B] (a11) [right=of a10] {};
    \draw (a0.north west) -- ++(-0.4cm,0)
          (a0.south west) -- ++(-0.4cm,0)
          (a11.north east) -- ++(0.4cm,0)
          (a11.south east) -- ++(0.4cm,0);


    \node[A] (b0) [below=of a0] {};
    \node[A] (b1) [right=of b0] {};
    \node[A] (b2) [right=of b1] {};
    \node[A] (b3) [right=of b2] {};
    \node[B] (b4) [right=of b3] {$\bullet$};
    \node[B] (b5) [right=of b4] {$\bullet$};
    \node[B] (b6) [right=of b5] {};
    \node[B] (b7) [right=of b6] {};
    \node[B] (b8) [right=of b7] {};
    \node[B] (b9) [right=of b8] {};
    \node[B] (b10) [right=of b9] {};
    \node[B] (b11) [right=of b10] {};
    \draw (b0.north west) -- ++(-0.4cm,0)
          (b0.south west) -- ++(-0.4cm,0)
          (b11.north east) -- ++(0.4cm,0)
          (b11.south east) -- ++(0.4cm,0);


    \node[A] (c0) [below=of b0] {};
    \node[A] (c1) [right=of c0] {};
    \node[A] (c2) [right=of c1] {};
    \node[A] (c3) [right=of c2] {};
    \node[A] (c4) [right=of c3] {};
    \node[B] (c5) [right=of c4] {$\bullet$};
    \node[B] (c6) [right=of c5] {$\bullet$};
    \node[B] (c7) [right=of c6] {};
    \node[B] (c8) [right=of c7] {};
    \node[B] (c9) [right=of c8] {};
    \node[B] (c10) [right=of c9] {};
    \node[B] (c11) [right=of c10] {};
    \draw (c0.north west) -- ++(-0.4cm,0)
          (c0.south west) -- ++(-0.4cm,0)
          (c11.north east) -- ++(0.4cm,0)
          (c11.south east) -- ++(0.4cm,0);


    \node[A] (d0) [below=of c0] {};
    \node[A] (d1) [right=of d0] {};
    \node[A] (d2) [right=of d1] {};
    \node[A] (d3) [right=of d2] {};
    \node[A] (d4) [right=of d3] {};
    \node[A] (d5) [right=of d4] {};
    \node[B] (d6) [right=of d5] {$\bullet$};
    \node[B] (d7) [right=of d6] {$\bullet$};
    \node[B] (d8) [right=of d7] {};
    \node[B] (d9) [right=of d8] {};
    \node[B] (d10) [right=of d9] {};
    \node[B] (d11) [right=of d10] {};
    \draw (d0.north west) -- ++(-0.4cm,0)
          (d0.south west) -- ++(-0.4cm,0)
          (d11.north east) -- ++(0.4cm,0)
          (d11.south east) -- ++(0.4cm,0);


    \node (Al) [draw=none] at (c1) {\huge $A$};
    \node (Bl) [draw=none] at (c8.north east) {\huge $B$};

  \end{tikzpicture}
  \caption{Simulating an SCA with low one-way communication complexity.
    (For simplicity, in this example the SCA does not shrink.)
    $B$ communicates the states of the cells marked with \enquote{$\bullet$}.
    The colors indicate which states are computed by each player.}
  \label{fig_sca_comm_compl_lem}
\end{figure}

We are now in position to prove Theorem~\ref{thm_linear_block_lang}.

\begin{proof}[Proof of Theorem~\ref{thm_linear_block_lang}]
  We prove that, for our language $L_1$ of before and $m(n) = n (n+1)$ (i.e.,
  $A_N$ receives the first $n$ input blocks),
  $\owComm^m(\Block_n(L_1))(N) \ge 2^n - n$.
  Since the input length is $N \in \Theta(n \cdot 2^n)$, the claim then follows
  from the contrapositive of Lemma~\ref{lem_comm_compl_sca}.

  The proof is by a counting argument.
  Let $A_N$ and $B_N$ be as in Definition~\ref{def_comm_compl}, and let $Y =
  \binalph^{2^n - n}$.
  The basic idea is that, for the same input $w_A$, if $B_N$ is given different
  inputs $w_B$ and $w_B'$ but $B_N(w_B) = B_N(w_B')$, then $w = w_A w_B$ is
  accepted if and only if $w' = w_A w_B'$ is accepted.
  Hence, for any $y, y' \in Y$ with $y \neq y'$, we must have $B_N(w_B) \neq
  B_N(w_B')$, where $w_B, w_B' \in \B^{2^n - n}_n$ are the block word versions
  of $y$ and $y'$, respectively; this is because, letting $j \in [0, 2^n-n]$ be
  such that $y(j) \neq y'(j)$ and $z = \bin_n(n+j)$, precisely one of the words
  $z y$ and $z y'$ is in $L_1$ (and the other not).
  Finally, note there is a bijection between $Y$ and the set $Y'$ of block words
  in $\B_n^{2^n - n}$ whose block numbering starts with $n + 1$ (i.e., $x_0 = n
  + 1$, where $x_0$  is as in Definition~\ref{def_block_word}) and with block
  entries of the form $a0^{n-1}$ where $a \in \binalph$ (i.e., $Y'$ is
  essentially the block version of $Y$ as in Definition~\ref{def_block_ver} but
  where we set $x_0 = n + 1$ instead of $x_0 = 0$).
  We conclude $\owComm^m(\Block_n(L_1))(N) \ge \abs{Y'} = \abs{Y} = 2^n - n$,
  and the claim follows.
\end{proof}



\section{Simulation of an SCA by a Streaming Algorithm}
\label{sec_sca_streaming}

In this section, we recall and prove:

\restatethmscatostreaming*

Before we state the proof, we first introduce some notation.
Having fixed an input $w$, let $c_t(i)$ denote the state of cell $i$ in step $t$
on input $w$.
Note that here we explicitly allow $c_t(i)$ to be the state $\otimes$ and also
disregard any changes in indices caused by cell deletion; that is, $c_t(i)$
refers to the \emph{same cell} $i$ as in the initial configuration $c_0$ (of
Definition~\ref{def_CA}; see also the discussion following
Definition~\ref{def_SCA}).
For a finite, non-empty $I = [a,b] \subseteq \Z$ and $t \in \N_0$, let
$\nndcl_t(I) = \max\{ i \mid i < a, c_t(i) \neq \otimes \}$ denote the nearest
non-deleted cell to the left of $I$; similarly,
$\nndcr_t(I) = \min\{ i \mid i > b, c_t(i) \neq \otimes \}$ is the nearest such
cell to the right of $I$.

\begin{proof}
  Let $S$ be an $O(t)$-time SCA for $L$.
  Using $S$, we construct a streaming algorithm $A$
  (Algorithm~\ref{alg_stream_sca}) for $L$ and prove it has the purported
  complexities.

  \begin{algorithm}[t]
    Compute $t(\abs{w})$\;
    Initialize lists $\mathtt{leftIndex}$, $\mathtt{centerIndex}$,
    $\mathtt{leftState}$, and $\mathtt{centerState}$\;
    $\mathtt{leftIndex}[0] \gets -1$;
    $\mathtt{leftState}[0] \gets q$\;
    $\mathtt{centerIndex}[0] \gets 0$;
    $\mathtt{centerState}[0] \gets w(0)$\;
    $\mathtt{next} \gets 1$\;
    $j_0 \gets 0$\;
    \For{$\tau \gets 0, \dots, t(\abs{w}) - 1$}{
      \nl\label{line_forloopstart}%
      $j \gets j_0$\;
      \nl\label{line_cond_read}%
      \uIf{$\mathtt{next} < \abs{w}$}{
        \nl\label{line_readnext}%
        $\mathtt{rightIndex} \gets \mathtt{next}$;
        $\mathtt{rightState} \gets w(\mathtt{next})$\;
        $\mathtt{next} \gets \mathtt{next} + 1$\;
      }
      \Else{
        \nl\label{line_incrementj0}%
        $\mathtt{rightIndex} \gets \abs{w}$;
        $\mathtt{rightState} \gets q$\;
        $j_0 \gets j_0 + 1$\;
      }
      \While{$j \le \tau$}{
        \nl\label{line_rprime_sprime}%
        $\mathtt{newRightIndex} \gets \mathtt{centerIndex}[j]$;
        $\mathtt{newRightState} \gets \delta(\mathtt{leftState}[j],
        \mathtt{centerState}[j], \mathtt{rightState})$\;
        $\mathtt{leftIndex}[j] \gets \mathtt{centerIndex}[j]$;
        $\mathtt{leftState}[j] \gets \mathtt{centerState}[j]$\;
        $\mathtt{centerIndex}[j] \gets \mathtt{rightIndex}$;
        $\mathtt{centerState}[j] \gets \mathtt{rightState}$\;
        $\mathtt{rightIndex} \gets \mathtt{newRightIndex}$;
        $\mathtt{rightState} \gets \mathtt{newRightState}$\;
        \nl\label{line_goto}%
        \lIf{$\mathtt{rightState} = \otimes$}{\Goto \ref{line_forloopstart}}
        $j \gets j + 1$\;
      }
      $\mathtt{leftIndex}[\tau+1] \gets -1$;
      $\mathtt{leftState}[\tau+1] \gets q$\;
      $\mathtt{centerIndex}[\tau+1] \gets \mathtt{rightIndex}$;
      $\mathtt{centerState}[\tau+1] \gets \mathtt{rightState}$\;
      \nl\label{line_acceptif}%
      \lIf{$\mathtt{centerState}[\tau+1] = a$}{\Accept}
    }
    \Reject\;
    \caption{Streaming algorithm $A$}
    \label{alg_stream_sca}
  \end{algorithm}

  \paragraph{Construction.}

  Let $w$ be an input to $A$.
  To decide $L$, $A$ computes the states of the cells of $S$ in the time steps
  up to $t(\abs{w})$.
  In particular, $A$ sequentially determines the state of the leftmost active
  cell in each of these time steps (starting from the initial configuration)
  and accepts if and only if at least one of these states is accepting.
  To compute these states efficiently, we use an approach based on dynamic
  programming, reusing space as the computation evolves.

  $A$ maintains lists $\mathtt{leftIndex}$, $\mathtt{leftState}$,
  $\mathtt{centerIndex}$, and $\mathtt{centerState}$ and which are indexed by
  every step $j$ starting with step zero and up to the current step $\tau$.
  The lists $\mathtt{leftIndex}$ and $\mathtt{centerIndex}$ store cell indices
  while $\mathtt{leftState}$ and $\mathtt{centerState}$ store the states of the
  respective cells, that is,
  $\mathtt{leftState}[j] = c_j(\mathtt{leftIndex}[j])$ and
  $\mathtt{centerState}[j] = c_j(\mathtt{centerIndex}[j])$.

  Recall the state $c_{j+1}(y)$ of a cell $y$ in step $j + 1$ is determined
  exclusively by the previous state $c_j(y)$ of $y$ as well as the states
  $c_j(x)$ and $c_j(z)$ of the left and right neighbors $x$ and $z$
  (respectively) of $y$ in the previous step $j$ (i.e., $x = \nndcl_j(y)$ and
  $z = \nndcr_j(y)$).
  In the variables maintained by $A$, $x$ and $c_j(x)$ correspond to
  $\mathtt{leftIndex}[j]$ and $\mathtt{leftState}[j]$, respectively, and $y$ and
  $c_j(y)$ to $\mathtt{centerIndex}[j]$ and $\mathtt{centerState}[j]$,
  respectively.
  $z$ and $c_j(z)$ are not stored in lists but, rather, in the variables
  $\mathtt{rightIndex}$ and $\mathtt{rightState}$ (and are determined
  dynamically).
  The cell indices computed (i.e., the contents of the lists
  $\mathtt{leftIndex}$ and $\mathtt{centerIndex}$ and the variables
  $\mathtt{rightIndex}$ and $\mathtt{newRightIndex}$) are not actually used by
  $A$ to compute states and are inessential to the algorithm itself; we use
  them only to simplify the proof of correctness below (and, hence, do not count
  them towards the space complexity of $A$).

  In each iteration of the \KwSty{for} loop, $A$ determines
  $c_{\tau+1}(z_0^\tau)$, where $z_0^\tau$ is the leftmost active cell of $S$
  in step $\tau$, and stores it $\mathtt{centerState}[\tau + 1]$.
  $\mathtt{next}$ is the index of the next symbol of $w$ to be read (or
  $\abs{w}$ once every symbol has been read), and $j_0$ is the minimal time step
  containing a cell whose state must be known to determine $c_{\tau+1}(z_0^t)$
  and remains $0$ as long as $\mathtt{next} < \abs{w}$.
  Hence, the termination of $A$ is guaranteed by the finiteness of $w$, that is,
  $\mathtt{next}$ can only be increased a finite number of times and, once all
  symbols of $w$ have been read (i.e., the condition in
  line~\ref{line_cond_read} no longer holds), by the increment of $j_0$ in
  line~\ref{line_incrementj0}.

  In each iteration of the \KwSty{while} loop, the algorithm starts from a local
  configuration in step $j$ of a cell $y = \mathtt{centerIndex}[j]$ with left
  neighbor $x = \mathtt{leftIndex}[j] = \nndcl_j(y)$ and right neighbor
  $z = \mathtt{rightIndex}[j] = \nndcl_j(y)$.
  It then computes the next state $c_{j+1}(y)$ of $y$ and sets $y$ as the new
  left cell and $z$ as the new center cell for step $j$.
  As long as it is not deleted (i.e., $c_{j+1}(y) \neq \otimes$), $y$ then
  becomes the right cell for step $j + 1$.
  In fact, this is the only place (line~\ref{line_goto}) in the algorithm where
  we need to take into consideration that $S$ is a shrinking (and not just a
  regular) CA.
  The strategy we follow here is to continue computing states of cells to the
  right of the current center cell (i.e., $y = \mathtt{centerIndex}[j]$) until
  the first cell to its right which has not deleted itself (i.e., $\nndcr_j(y)$)
  is found.
  With this non-deleted cell we can then proceed with the computation of the
  state of $\mathtt{centerIndex}[j+1]$ in step $j + 1$.
  Hence, if $y$ has deleted itself, to compute the state of the next cell to
  its right we must either read the next symbol of $w$ or, if there are no
  symbols left, use quiescent cell number $\abs{w}$ as right neighbor in step
  $j_0$, computing states up until we are at step $j$ again (hence the
  \KwSty{goto} instruction).

  \paragraph{Correctness.}

  The following invariants hold for both loops in $A$:
  \begin{enumerate}
    \item $\mathtt{centerIndex}[\tau] = \min\{ z \in \N_0 \mid c_\tau(z) \neq
    \otimes \}$,
    that is, $\mathtt{centerIndex}[\tau]$ is the leftmost active cell of $S$ in
    step $j$.
    \item If $j \le \tau$, then
    $\mathtt{rightIndex} = \nndcr_j(\mathtt{centerIndex}[j])$ and
    $\mathtt{rightState} = c_j(\mathtt{rightIndex})$.
    \item For every $j' \in [j_0, \tau]$:
    \begin{itemize}
      \item $\mathtt{leftIndex}[j'] = \nndcl_{j'}(\mathtt{centerIndex}[j'])$,
      \item $\mathtt{leftState}[j'] = c_{j'}(\mathtt{leftIndex}[j'])$; and
      \item $\mathtt{centerState}[j'] = c_{j'}(\mathtt{centerIndex}[j'])$.
    \end{itemize}
  \end{enumerate}
  These can be shown together with the observation that, following the
  assignment of $\mathtt{newRightIndex}$ and $\mathtt{newRightState}$ in
  line~\ref{line_rprime_sprime}, we have
  $\mathtt{newRightState} = c_{j+1}(\mathtt{newRightIndex})$ and, in case
  $\mathtt{newRightState} \neq \otimes$ and $j < \tau$, then also
  $\mathtt{newRightIndex} = \nndcr_j(\mathtt{centerIndex}[j+1])$.
  Using the above, it follows that after the execution of the \KwSty{while} loop
  we have $j = \tau + 1$, $\mathtt{rightState} \neq \otimes$, and
  $\mathtt{rightState} = c_{\tau + 1}(\mathtt{rightIndex})$.
  Since then $\mathtt{rightIndex} = \mathtt{centerIndex}[j-1] =
  \mathtt{centerIndex}[\tau]$, we obtain
  $\mathtt{rightIndex} = \min\{ z \in \N_0 \mid c_{\tau+1}(z) \neq \otimes \}$.
  Hence, as $\mathtt{centerState}[\tau + 1] = \mathtt{rightState} =
  c_{\tau+1}(\mathtt{rightIndex})$ holds in
  line~\ref{line_acceptif}, if $A$ then accepts, so does $S$ accept $w$ in step
  $\tau$.
  Conversely, if $A$ rejects, then $S$ does not accept $w$ in any step
  $\tau \le t(\abs{w})$.

  \paragraph{Complexity.}

  The space complexity of $A$ is dominated by the lists $\mathtt{leftState}$ and
  $\mathtt{centerState}$, which has $O(t(\abs{w}))$ many entries of $O(1)$ size.
  As mentioned above, we ignore the space used by the lists $\mathtt{leftIndex}$
  and $\mathtt{centerIndex}$ and the variables $\mathtt{rightIndex}$ and
  $\mathtt{newRightIndex}$ since they are inessential (i.e., if we remove them
  as well as all instructions in which they appear, the algorithm obtained is
  equivalent to $A$).

  As for the update time, note each list access or arithmetic operation costs
  $O(\log t(\abs{w}))$ time (since $t(\abs{w})$ upper bounds all numeric
  variables).
  Every execution of the \KwSty{while} loop body requires then
  $O(\log t(\abs{w}))$ time and, since, there are at most $O(t(\abs{w}))$
  executions between any two subsequent reads (i.e., line~\ref{line_readnext}),
  this gives us the purported $O(t(\abs{w}) \log t(\abs{w}))$ update time.

  Finally, for the reporting time of $A$, as soon as $i = \abs{w}$ holds after
  execution of line~\ref{line_readnext} (i.e., $A$ has completed reading its
  input) we have that the \KwSty{while} loop body is executed at most
  $\tau - j + 1$ times before line~\ref{line_readnext} is reached again.
  Every time this occurs (depending on whether line~\ref{line_readnext} is
  reached by the \KwSty{goto} instruction or not), either $j_0$ or both $j_0$
  and $\tau$ are incremented.
  Hence, since $\tau \le t(\abs{w})$, we have an upper bound of
  $O(t(\abs{w})^2)$ executions of the \KwSty{while} loop body, resulting (as
  above) in an $O(t(\abs{w})^2 \log t(\abs{w}))$ reporting time in total.
\end{proof}


\section{Hardness Magnification for Sublinear-Time SCAs}
\label{sec_sca_hm}

Let $K > 0$ be constant such that, for any function $s\colon \N_+ \to \N_+$,
every circuit of size at most $s(n)$ can be described by a binary string of
length at most $\ell(n) = K s(n) \log s(n)$.
In addition, let $\bot$ denote a string (of length at most $\ell(n)$) such that
no circuit of size at most $s(n)$ has $\bot$ as its description.
Furthermore, let $\Merge[s]$ denote the following search problem (adapted from
\cite{mckay19_weak_stoc}):
\begin{description}
  \item[Given:]
    the binary representation of $n \in \N_+$, the respective descriptions
    (padded to length $\ell(n)$) of circuits $C_0$ and $C_1$ such that
    $\abs{C_i} \le s(n)$, and $\alpha, \beta, \gamma \in \binalph^n$ with
    $\alpha \le \beta \le \gamma < 2^n$.
  \item[Find:]
    the description of a circuit $C$ with $\abs{C} \le s(n)$ and such that
    $\forall x \in [\alpha, \beta-1]: C(x) = C_0(x)$ and
    $\forall x \in [\beta, \gamma-1]: C(x) = C_1(x)$; if no such $C$ exists or
    $C_i = \bot$ for any $i$, answer with $\bot$.
\end{description}
Note that the decision version of $\Merge[s]$, that is, the problem of
determining whether a solution to an instance $\Merge[s]$ exists is in
$\Sigma_2^p$.
Moreover, $\Merge[s]$ is Turing-reducible (in polynomial time) to a decision
problem very similar to $\Merge[s]$ and which is also in $\Sigma_2^p$, namely
the decision version of $\Merge[s]$ but with the additional requirement that the
description of $C$ admits a given string $v$ of length $\abs{v} \le s(n)$ as a
prefix.%
\footnote{This is a fairly common construction in complexity theory for
  reducing search to decision problems; refer to
  \cite{goldreich08_computational_book} for the same idea applied in other
  contexts.
}

We now formulate our main theorem concerning SCAs and MCSP:
\begin{theorem} \label{thm_sca_hm_explicit}
  Let $s\colon \N_+ \to \N_+$ be constructible in place by a CA in $O(s(n))$
  time.
  Furthermore, let $m = m(n)$ denote the maximum instance length of
  $\Merge[s]$, and let $f, g\colon \N_+ \to \N_+$ with $f(m) \ge g(m) \ge m$
  be constructible in place by a CA in $O(f(m))$ time and $O(g(m))$ space.
  Then, for $b(n) = \floor{g(m) / 2}$, if $\Merge[s]$ is computable in place by
  a CA in at most $f(m)$ time and $g(m)$ space, then the search version of
  $\Block_b(\MCSP[s])$ is computable by an SCA in $O(n \cdot f(m))$ time, where
  the instance size of the latter is in $\Theta(2^n \cdot b(n))$.
\end{theorem}

We are particularly interested in the repercussions of
Theorem~\ref{thm_sca_hm_explicit} \emph{taken in the contrapositive}.
Since $\P = \NP$ implies $\P = \Sigma_2^p$, it also implies there is a poly-time
Turing machine for $\Merge[s]$; since a CA can simulate a Turing machine with no
time loss, for $m$ as above we obtain:

\restatethmscahm*

We now turn to the proof of Theorem~\ref{thm_sca_hm_explicit}, which follows
\cite{mckay19_weak_stoc} closely.
First, we generalize blockwise reductions (see Definition~\ref{def_bw_red}) to
search problems:

\begin{definition}[Blockwise reducible (for search problems)]
  \label{def_bw_red_search}
  Let $L$ and $L'$ be block languages that correspond to search problems $S$ and
  $S'$, respectively.
  Also, for an instance $x$, let $S(x)$ (resp., $S'(x)$) denote the set of
  solutions for $x$ under the problem $S$ (resp., $S'$).
  Then $L$ is said to be (\emph{$k$-})\emph{blockwise reducible} to $L'$ if
  there is a computable $k$-blockwise map $g\colon \B^{km}_b \to \B^m_b$ such
  that, for every $w \in \B^{km}_b$, we have $S(w) = S'(g(w))$.
\end{definition}

Notice Lemma~\ref{lem_bw_reduction_chain} readily generalizes to blockwise
reductions in this sense.

Next, we describe the set of problems that we shall reduce $\Block_b(\MCSP[s])$
to.
Let $r\colon \N_+ \to \N_+$ be a function.
There is a straightforward $1$-blockwise reduction from $\Block_b(\MCSP[s])$ to
(a suitable block version of) the following search problem $\Merge_r[s]$:
\begin{description}
  \item[Given:]
  the binary representation of $n \in \N_+$ and the respective descriptions
  (padded to length $\ell(n)$) of circuits $C_1, \dots, C_r$, where
  $\abs{C_i} \le s(n)$ for every $i$ and $r = r(n)$.
  \item[Find:]
  (the description of) a circuit $C$ with $\abs{C} \le s(n)$ and such that,
  for every $i$ and every $x \in [(i-1) \cdot 2^n / r, i \cdot 2^n / r - 1]$,
  $C(x) = C_i(x)$; if no such $C$ exists or $C_i = \bot$ for any $i$, answer
  with $\bot$.
\end{description}
In particular, for the reduction mentioned above, we shall use $r = 2^n$.
Evidently, $\Merge_r[s]$ is a generalization of the problem $\Merge[s]$ defined
previously and, more importantly, every instance of $\Merge_r[s]$ is simply a
concatenation of $r / 2$ many $\Merge[s]$ instances where $\alpha$, $\beta$, and
$\gamma$ are given implicitly.
Using the assumption that $\Merge[s]$ is computable by a CA in at most $f(m)$
time and $g(m)$ space, we can solve each such instance in parallel, thus
producing an instance of $\Merge_{r/2}[s]$ (i.e., halving $r$).
This yields a $2$-blockwise reduction from (the respective block versions of)
$\Merge_r[s]$ to $\Merge_{r/2}[s]$ (cnf.~the proof of
Proposition~\ref{prop_sca_par_etc}).
Using Lemma~\ref{lem_bw_reduction_chain} and that $\Merge_1[s]$ is trivial, we
obtain the purported SCA for $\Block_b(\MCSP[s])$.

\begin{proof}
  Let $n$ be fixed, and let $r = 2^n$.
  First, we describe the $1$-blockwise reduction from $\Block_b(\MCSP[s])$ to a
  block version of $\Merge_r[s]$ (which we shall describe along with the
  reduction).
  Let $T_a$ denote the (description of the) trivial circuit that is constant
  $a \in \binalph$, that is, $T_a(x) = a$ for every $x \in \binalph^n$.
  Then we map each block $\binom{\bin_n(x)}{y0^{b(n) - 1}}$ with
  $y \in \binalph$ to the block $\binom{\bin_n(x)}{T_y \pi}$, where
  $\pi \in \{ 0 \}^\ast$ is a padding string so that the block length $b(n)$ is
  preserved.
  (This is needed to ensure enough space is available for the construction; see
  the details further below.)
  It is evident this can be done in time $O(b(n))$ and (since we just translate
  the truth-table $0$ and $1$ entries to the respective trivial circuits) that
  the reduction is correct, that is, that every solution to the original
  $\Block_b(\MCSP[s])$ instance must also be a solution of the produced instance
  of (the resulting block version of) $\Merge_r[s]$ and vice-versa.

  Next, maintaining the block representation described above, we construct the
  $2$-blockwise reduction from the respective block versions of $\Merge_\rho[s]$
  to $\Merge_{\rho/2}[s]$, where $\rho = 2^k$ for some $k \in [1,n]$.
  Let $A$ denote the CA that, by assumption, computes a solution to an instance
  of $\Merge[s]$ in place in at most $f(m)$ time and $g(m)$ space.
  Then, for $j \in [0, \rho / 2 - 1]$, we map each pair
  $\binom{\bin_n(2j)}{C_0 \pi_0} \# \binom{\bin_n(2j+1)}{C_1 \pi_1}$ of blocks
  (where $\pi_0, \pi_1 \in \{ 0 \}^\ast$ again are padding strings) to
  $\binom{\bin_n(j)}{C \pi}$, where $\pi \in \{ 0 \}^\ast$ is a padding string
  (as above) and $C$ is the circuit produced by $A$ for
  $\alpha = 2j \cdot 2^n / \rho$, $\beta = (2j + 1) \cdot 2^n / \rho$,
  and $\gamma = (2j + 2) \cdot 2^n / \rho$.

  To actually execute $A$, we need $g(m)$ space (which is guaranteed by the
  block length $b(n)$) and, in addition, to prepare the input so it is in the
  format expected by $A$ (i.e., eliminating the padding between the two circuit
  descriptions and writing the representations of $\alpha$, $\beta$, and
  $\gamma$), which can be performed in
  $O(b(n)) \subseteq O(g(m)) \subseteq O(f(m))$ time.
  For the correctness, suppose the above reduces an instance of $\Merge_\rho[s]$
  with circuits $C_1, \dots, C_\rho$ to an instance of $\Merge_{\rho/2}[s]$
  with circuits $D_1, \dots, D_{\rho/2}$ (and no $\bot$ was produced).
  Then, a circuit $E$ is a solution to the latter if and only if $E(x) = D_i(x)$
  for every $i$ and
  $x \in [(i - 1) \cdot 2^n / (\rho / 2), i \cdot 2^n / (\rho / 2) - 1]$.
  Using the definition of $\Merge[s]$, every $D_i$ must satisfy
  $D_i(x) = C_{2i-1}(x)$ and $D_i(y) = C_{2i}(y)$ for
  $x \in [(2i - 2) \cdot 2^n / \rho, (2i - 1) \cdot 2^n / \rho - 1]$ and
  $y \in [(2i - 1) \cdot 2^n / \rho, 2i \cdot 2^n / \rho - 1]$.
  Hence, $E$ agrees with $C_1, \dots, C_\rho$ if and only if it agrees with
  $D_1, \dots, D_{\rho/2}$ (on the respective intervals).

  Since $s(n) \ge n$ and $\Merge_1[s]$ is trivial (i.e., it can be accepted in
  $O(b(n))$ time), applying the generalization of
  Lemma~\ref{lem_bw_reduction_chain} to blockwise reductions for search problems
  completes the proof.
\end{proof}

\paragraph{Comparison with \cite{mckay19_weak_stoc}.}
We conclude this section with a comparison of our result and proof with
\cite{mckay19_weak_stoc}.
The most evident difference between the statements of
Theorems~\ref{thm_sca_hm}~and~\ref{thm_sca_hm_explicit} and the related result
from \cite{mckay19_weak_stoc} (i.e., Theorem~\ref{thm_mmw19}) is that our
results concern CAs (instead of Turing machines) and relate more explicitly to
the time and space complexities of $\Merge[s]$; in particular, the choice of the
block length is tightly related with the space complexity of computing
$\Merge[s]$.
As for the proof, notice that we only merge two circuits at a time, which makes
for a smaller instance size $m$ (of $\Merge[s]$); this not only simplifies the
proof but also minimizes the resulting time complexity of the SCA (as $f(m)$ is
then smaller).
Also, in our case, we make no additional assumptions regarding the first
reduction from $\Block_b(\MCSP[s])$ to $\Merge_r[s]$; in fact, this step can be
performed unconditionally.
Finally, we note that our proof renders all blockwise reductions explicit and
the connection to the self-reductions of \cite{allender10_amplifying_jacm} more
evident.
Despite these simplifications, the argument extends to generalizations of MCSP
with similar structure and instance size (e.g., MCSP in the setting of Boolean
circuits with oracle gates as in \cite{mckay19_weak_stoc} or MCSP for
multi-output functions as in \cite{ilango20_np-hardness_ccc}).


\section{Concluding Remarks}
\label{sec_conclusion}

\paragraph{Proving SCA Lower Bounds for $\MCSP[s]$.}
Recalling the language $L_1$ from the proof of
Theorem~\ref{thm_linear_block_lang}, consider the intersection $L_1[s] = L_1
\cap \MCSP[s]$.
Evidently, $L_1[s]$ is comparable in hardness to $\MCSP[s]$ (e.g., it is
solvable in polynomial time using a single adaptive query to $\MCSP[s]$).
By adapting the construction from the proof of Theorem~\ref{thm_sca_hm_explicit}
so the SCA additionally checks the $L_1$ property at the end in $\poly(s(n))$
time (e.g., using the circuit $C$ produced to check whether $C(x) = 1$ for $x =
C(0) \cdots C(n-1)$), we can derive a hardness magnification result for
$L_1[s]$ too:
If $\Block_b(L_1[s]) \not\in \SCA[\poly(s(n))]$ (for every $b \in \poly(s(n))$),
then $\P \neq \NP$.
Using the methods from Section~\ref{sec_linear_block_lang} and that there are
$2^{\Omega(s(n))}$ many (unique) circuits of size $s(n)$ or less,\footnote{%
  Let $K > 0$ be constant such that every Boolean function on $m$
  variables admits a circuit of size at most $K \cdot 2^m / m$.
  Setting $m = \floor{\log s(n)}$, notice that, for sufficiently large $n$ (and
  $s(n) \in \omega(1) \cap O(2^n / n)$), this gives us
  $s(n) \ge K \cdot 2^m / m$, thus implying that every Boolean function on
  $m \le n$ variables admits a circuit of size at most $s(n)$.
  Since there are $2^{2^m}$ many such (unique) functions, it follows there are
  $2^{\Omega(s(n))}$ (unique) circuits of size at most $s(n)$.
} this means that, if $\Block_b(L_1[s]) \in \SCA[t(n)]$ for some $b \in
\poly(n)$ and $t\colon \N_+ \to \N_+$, then $t \in \Omega(s(n))$.
Hence, for an eventual proof of $\P \neq \NP$ based on
Theorem~\ref{thm_sca_hm}, one would need to develop new techniques (see also
the discussion below) to raise this bound at the very least beyond
$\poly(s(n))$.

Seen from another angle, this demonstrates that, although we can prove a tight
SCA worst-case lower bound for $L_1$ (Theorem~\ref{thm_linear_block_lang}),
establishing similar lower bounds on instances of $L_1$ with low circuit
complexity (i.e., instances which are also in $\MCSP[s]$) is at least as hard as
showing $\P \neq \NP$.
In other words, it is straightforward to establish a lower bound for $L_1$ using
arbitrary instances, but it is absolutely non-trivial to establish similar lower
bounds for \emph{easy} instances of $L_1$ where instance hardness is measured in
terms of circuit complexity.

\paragraph{The Proof of Theorem~\ref{thm_sca_hm_explicit} and the Locality
Barrier.}
In a recent paper \cite{chen20_beyond_itcs}, \citeauthor{chen20_beyond_itcs}
propose the concept of a \emph{locality barrier} to explain why current lower
bound proof techniques (for a variety of non-uniform computational models) do
not suffice to show the lower bounds needed for separating complexity classes
in conjunction with hardness magnification (i.e., in our case above a
$\poly(s(n))$ lower bound that proves $\P \neq \NP$).
In a nutshell, the barrier arises from proof techniques relativizing with
respect to \emph{local aspects} of the computational model at hand (in
\cite{chen20_beyond_itcs}, concretely speaking, oracle gates of small fan-in),
whereas it is known that a proof of $\P \neq \NP$ must not relativize
\cite{baker75_relativizations_siamjc}.

The proof of Theorem~\ref{thm_sca_hm_explicit} confirms the presence of such a
barrier also in the uniform setting and concerning the separation of $\P$ from
$\NP$.
Indeed, the proof mostly concerns the construction of an SCA where the overall
computational paradigm of blockwise reductions (using
Lemma~\ref{lem_bw_reduction_chain}) is unconditionally compatible with the
SCA model (as exemplified in Proposition~\ref{prop_sca_par_etc}); the $\P = \NP$
assumption is needed exclusively so that the local algorithm for $\Merge[s]$
in the statement of the theorem exists.
Hence, the result also holds \emph{unconditionally} for SCAs that are, say,
augmented with oracle access (in a plausible manner, e.g., by using an
additional oracle query track and special oracle query states) to $\Merge[s]$.
(Incidentally, the same argument also applies to the proof of the hardness
magnification result for streaming algorithms (i.e., Theorem~\ref{thm_mmw19}) in
\cite{mckay19_weak_stoc}, which also builds on the existence of a similar
locally computable function.)
In particular, this means the lower bound techniques from the proof of
Theorem~\ref{thm_linear_block_lang} do not suffice since they extend to SCAs
having oracle access to any computable function.

\paragraph{Open Questions.}
We conclude with a few open questions:
\begin{itemize}
  \item By weakening SCAs in some aspect, certainly we can establish an
    unconditional MCSP lower bound for the weakened model which, were it to
    hold for SCAs, would imply the separation $\P \neq \NP$ (using
    Theorem~\ref{thm_sca_hm}).
    \emph{What forms of weakening} (conceptually speaking) are needed for these
    lower bounds?
    How are these related to the locality barrier discussed above?
  \item Secondly, we saw SCAs are strictly more limited than streaming
    algorithms.
    Proceeding further in this direction, can we identify
    \emph{further (natural) models of computation that are more restricted than
    SCAs} (whether CA-based or not) and for which we can prove results similar
  to Theorem~\ref{thm_sca_hm_explicit}?
  \item Finally, besides MCSP, what other (natural) problems admit similar SCA
    hardness magnification results?
    More importantly, can we identify some \emph{essential property} of these
    problems that would explain these results?
    For instance, in the case of MCSP there appears to be some connection to the
    length of (minimal) witnesses being much smaller than the instance length.
    Indeed, one sufficient condition in this sense (disregarding SCAs) is
    sparsity \cite{chen19_hardness_focs}; nevertheless, it seems rather
    implausible that this would be the sole property responsible for all
    hardness magnification phenomena.
\end{itemize}

\paragraph{Acknowledgments.}
\myack


\printbibliography

\end{document}